\documentclass[11pt,a4paper]{article}

\usepackage{array,epsfig}
\usepackage{amsthm}
\usepackage{amsmath}
\usepackage{amsfonts}
\usepackage{amssymb}
\usepackage{amsxtra}
\usepackage{enumitem}
\usepackage{mathrsfs}
\usepackage{datetime}
\usepackage{color}
\usepackage{mathtools}
\usepackage[noend]{algpseudocode}
\usepackage{tcolorbox}
\usepackage{graphicx}
\usepackage[english]{babel}  % force American English hyphenation patterns
\usepackage[utf8x]{inputenc}  % unicode
\usepackage{wrapfig}
\usepackage{pifont}
\usepackage{listings}
\usepackage{graphicx}
\usepackage{pgfplots}
\usepackage{hyperref}
\usepackage{footnote}
\usepackage[font=small]{subfig}
\usepackage{pgfplots}
\usepackage[linesnumbered,ruled,boxed,vlined]{algorithm2e}
\usepackage{sistyle}
\usepackage{cite}
\usepackage{natbib}
\usepackage{pifont}
\usepackage{authblk}

\algnewcommand{\IfThen}[2]{\algorithmicif\ #1\ \algorithmicthen\ #2}

\newcommand{\exact}{{\textsc{ExactBFC}}}
\newcommand{\pervertex}{{\textsc{vBFC}}}
\newcommand{\peredge}{{\textsc{eBFC}}}

\newcommand{\wangExact}{{\textsc{WFC}}}

\SIthousandsep{,}

\pgfplotsset{compat=1.9} % update to 1.12 if possible!
\pgfplotsset{%
    axis line origin/.style args={#1,#2}{
        x filter/.append code={ % Check for empty or filtered out numbers
            \ifx\pgfmathresult\empty\else\pgfmathparse{\pgfmathresult-#1}\fi
        },
        y filter/.append code={
            \ifx\pgfmathresult\empty\else\pgfmathparse{\pgfmathresult-#2}\fi
        },
        xticklabel=\pgfmathparse{\tick+#1}\pgfmathprintnumber{\pgfmathresult},
        yticklabel=\pgfmathparse{\tick+#2}\pgfmathprintnumber{\pgfmathresult}
    }
}
\definecolor{urlclr}{HTML}{0000FF}
\definecolor{bblue}{HTML}{4F81BD}
\definecolor{rred}{HTML}{C0504D}
\definecolor{ggreen}{HTML}{199400}
\definecolor{ppurple}{HTML}{8900ff}
\definecolor{codegreen}{HTML}{094c18}
\definecolor{codegray}{rgb}{9.5,9.5,9.5}
\definecolor{codepurple}{rgb}{0.58,0,0.82}
\definecolor{backcolour}{rgb}{0.95,0.95,0.92}
\hypersetup{
    colorlinks=true,
    linkcolor=urlclr,
    filecolor=magenta,
    urlcolor= ggreen,
    citecolor= ppurple,
    breaklinks=true,
}
\newtheorem{lem}{Lemma}

\newcommand{\bigO}{\ensuremath{\mathcal{O}}}
\newcommand{\Var}{\mathbb{V}\mathrm{ar}}
\newcommand{\Cov}{\mathbb{C}\mathrm{ov}}
\newcommand{\E}{\mathbb{E}}
\usepackage[nameinlink,capitalise,noabbrev]{cleveref}
\crefname{lem}{Lemma}{Lemmas}
\crefname{line}{Line}{Lines}
\crefname{obs}{Observation}{Observations}

\newcommand{\expec}[1]{{\mathbb E}\left [ #1 \right ]}
\newcommand{\prob}[1]{{\Pr}\left [ #1 \right ]}
\newcommand{\var}[1]{\mathbb{V}\mathrm{ar}\left [ #1 \right ]}

\newcommand{\prt}[1]{\left( #1 \right)}
\newcommand{\wdg}{\land}
\newcommand{\lbr}{\left(}
\newcommand{\rbr}{\right)}
\newcommand{\deli}{\texttt{Deli}}
\newcommand{\web}{\texttt{Web}}
\newcommand{\ork}{\texttt{Orkut}}
\newcommand{\jrn}{\texttt{Journal}}
\newcommand{\wiki}{\texttt{Wiki-en}}
\newcommand{\twt}{\texttt{Twitter}}

\newcommand{\clr}{\textsc{ClrSpar}{}}
\newcommand{\edgSpr}{\textsc{ESpar}{}}
\newcommand{\fstEdg}{\textsc{Fast-eBFC}{}}
\newcommand{\edgSam}{\textsc{ESamp}{}}
\newcommand{\wdgSam}{\textsc{WSamp}{}}
\newcommand{\verSam}{\textsc{VSamp}{}}

\renewcommand{\triangle}{\Delta}

\usepackage{mathtools,pict2e,epic,relsize}

\newlength\ztimespadding
\newcommand*\btrflyaux[2]{\vcenter{\hbox{\hspace{\ztimespadding}%
  \ifx#1\displaystyle
    \setlength\unitlength{1.5ex}%
    \linethickness{.1ex}%
  \else\ifx#1\textstyle
    \setlength\unitlength{1.5ex}%
    \linethickness{.1ex}%
  \else\ifx#1\scriptstyle
    \setlength\unitlength{1.2ex}%
    \linethickness{.1ex}%
  \else\ifx#1\scriptscriptstyle
    \setlength\unitlength{.975ex}%
    \linethickness{.1ex}%
  \fi\fi\fi\fi
  \begin{picture}(1,1)\roundcap
    \put(0,0.27){\line(1,0){1}}
    \put(1,0.27){\line(-1,1){1}}
    \put(0,1.27){\line(1,0){1}}
    \put(0,0.27){\line(1,1){1}}
  \end{picture}%
  \hspace{\ztimespadding}%
      \vspace{-0.6ex}
}}}

\newcommand*\bfly{\DOTSB\mathbin{\mathpalette\btrflyaux\relax}}
\DeclareFontFamily{U} {MnSymbolC}{}

\DeclareFontShape{U}{MnSymbolC}{m}{n}{
  <-6> MnSymbolC5
  <6-7> MnSymbolC6
  <7-8> MnSymbolC7
  <8-9> MnSymbolC8
  <9-10> MnSymbolC9
  <10-12> MnSymbolC10
  <12-> MnSymbolC12}{}
\DeclareFontShape{U}{MnSymbolC}{b}{n}{
  <-6> MnSymbolC-Bold5
  <6-7> MnSymbolC-Bold6
  <7-8> MnSymbolC-Bold7
  <8-9> MnSymbolC-Bold8
  <9-10> MnSymbolC-Bold9
  <10-12> MnSymbolC-Bold10
  <12-> MnSymbolC-Bold12}{}

\DeclareSymbolFont{MnSyC} {U} {MnSymbolC}{m}{n}
\DeclareMathSymbol{\bbbfly}{\mathbin}{MnSyC}{39}

\usepackage{multirow, makecell}
\usepackage{diagbox}
\newcommand{\myremove}[1]{}

\newcommand{\erdem}[1]{\textcolor{red}{(#1)}}

\title{Butterfly Counting in Bipartite Networks}

\author[1]{Seyed-Vahid Sanei-Mehri\thanks{vas@iastate.edu}}
\author[2]{Erdem Sar{\i}y\"{u}ce\thanks{erdem@Buffalo.edu}}
\author[1]{Srikanta Tirthapura\thanks{snt@iastate.edu}}
\affil[1]{Iowa State University}
\affil[2]{University at Buffalo}

\begin{document}
\date{}
\maketitle
\vspace{-40px}
\begin{abstract}
We consider the problem of counting motifs in bipartite affiliation networks, such as author-paper, user-product, and actor-movie relations. We focus on counting the number of occurrences of a ``butterfly'', a complete $2 \times 2$ biclique, the simplest cohesive higher-order structure in a bipartite graph. Our main contribution is a suite of randomized algorithms that can quickly approximate the number of butterflies in a graph with a provable guarantee on accuracy. An experimental evaluation on large real-world networks shows that our algorithms return accurate estimates within a few seconds, even for networks with trillions of butterflies and hundreds of millions of edges. %The substantial prior work on counting motifs in unipartite graphs, such as counting triangles, does not apply here, since bipartite graphs do not have triangles. 
\end{abstract}

\section{Introduction}
\label{sec:intro}

Graph motifs are used to model and examine interactions among small sets of vertices in networks. Finding frequent patterns of interactions can reveal functions of participating entities~\cite{SePiKo14, Jha15, Pinar17, Ahmed16, Bressan17, Jain17} and help characterize the network. Also known as graphlets or higher-order structures, motifs are regarded as basic building blocks of complex networks in domains such as social networks, food webs, and neural networks~\cite{Milo02}.
For this reason, finding and counting motifs are among the most important and widely used network analysis procedures. The triangle is the most basic motif in a unipartite network, and graph mining literature is abundant with triangle counting algorithms for stationary networks~\cite{SePiKo14,Babis09} as well as dynamic networks~\cite{Lim15, DeStefani16, PTTW13}. There are also studies that consider structures with more than 3 vertices~\cite{Jha15,Pinar17,Jain17}, but these primarily focus on cliques, such as 4-cliques and 5-cliques.
% and 5-vertex patterns~\cite{Pinar17}, and small cliques of size up to 10~\cite{Jain17}.

In this work, we focus on {\em bipartite (affiliation) networks}, an important type of network for many applications~\cite{Borgatti97, Latapy08}.
For example, relationships between authors and papers can be modeled as a bipartite graph, where authors form one vertex partition, papers form the other vertex partition, and an author has an edge to each paper that she published. 
Other examples include user-product relations, word-document affiliations, and actor-movie networks.
Bipartite graphs can represent hypergraphs that capture many-to-many relations among entities. %, that are not necessarily pairwise.
A hypergraph $H=(V_H,E_H)$ with vertex set $V_H$ and edge set $E_H$, where each hyper-edge $h \in E_H$ is a set of vertices, can be represented as a 
bipartite graph with vertex set $V_H \cup E_H$ with one partition for $V_H$ and another for $E_H$, and an edge from a vertex $v \in V_H$ to an edge $h \in E_H$ if $v$ is a part of $h$ in $H$.
For example, a hypergraph corresponding to an author-paper relation where each paper is associated with a set of authors can be represented using a bipartite graph with one partition for authors and one for papers.

%for a paper with ten authors, a bipartite graph enables to associate a set of ten authors with the same paper.
%in the same relation (paper).

A common approach to handle a bipartite network is to reduce it to a unipartite co-occurrence network by a projection~\cite{Newman01a, Newman01b}. A projection selects a vertex partition as the set of entities, and creates a unipartite network whose vertex set is the set of all entities and where two entities are connected if they share an affiliation in the bipartite network. For the author-paper network, a projection on the authors creates a unipartite co-authorship network. However, such a projection causes the number of edges in the graph to explode, artificially boosts the number of triangles and clustering coefficients, and results in information loss. For instance, we observed up to 4 orders of magnitude increase in size when the bipartite network between wikipedia articles and their editors in French is projected onto a unipartite network of articles -- number of edges goes from $22$M to more than $200$B. As a result, it is preferable to analyze bipartite networks directly.

While there is extensive work on motif counting in unipartite networks, these do not apply to bipartite networks. Motifs in bipartite networks are very different from motifs in a unipartite network. The most commonly studied motifs in a unipartite network are cliques of small sizes, but a bipartite graph does not have any cliques with more than two vertices, not even a triangle! Instead, natural motifs in a bipartite network are bicliques of small size.%, as we elaborate further.

The most basic motif that models cohesion in a bipartite network is the complete $2 \times 2$ biclique, also known as a \textbf{butterfly}~\cite{Aksoy16,Sariyuce18} or a rectangle~\cite{Wang14}. Although there have been attempts at defining other cohesive motifs in bipartite networks, such as the complete $3 \times 3$ biclique~\cite{Borgatti97} and $4$-path~\cite{Opsahl13}, the butterfly remains the smallest unit of cohesion, and has been used in defining basic metrics such as the clustering coefficient in a bipartite graph~\cite{Robins04,Lind05}. In particular, it is the smallest subgraph that has multiple vertices on each side with multiple common neighbors. It can be considered as playing the same role in bipartite networks as the triangle did in unipartite networks -- a building block for community structure. We aim to derive methods which can accurately estimate the number of butterflies in a given large bipartite graph. We view butterfly counting as a first, but important step towards general methods for motif counting and analysis of bipartite affiliation networks.

{\bf Contributions:} We present fast algorithms to accurately estimate the number of butterflies in a bipartite network. Our algorithms are simple to implement, backed up by theoretical guarantees, and have good practical performance.

\begin{itemize}[leftmargin=2ex]
\item \textbf{Exact Butterfly Counting:} We first present an efficient exact algorithm, \exact, for counting the number of butterflies in a network.
We use a simple measure, the sum of the squares of vertex degrees, to choose which vertex partition of the bipartite graph to start the algorithm from.
Leveraging the imbalance between vertex partitions yields significant speedups over the state-of-the-art~\cite{Wang14}.
%We also use this heuristic to count the butterflies per-vertex and per-edge \erdem{oh, we don't use that heuristic in local algs}.

\item \textbf{Randomized Approximate Butterfly Counting:} We introduce the first randomized algorithms to find the approximate number of butterflies in a network by sampling. Our algorithms are able to derive accurate estimates with error as low as $1$\% within a few seconds, are much faster than the exact algorithms, and have an insignificant memory print.
%Those are particularly valuable as the graph becomes larger, since the cost of the exact algorithm grows super-linearly with the graph size, while the cost of the randomized algorithms is much smaller. 
We present two types of randomized algorithms.
\begin{itemize}
\item {\bf One-shot Sparsification} techniques assume that the entire graph is available for processing. They thin-down the graph to a much smaller sparsified graph through choosing each edge of the original graph with a certain randomized procedure.  Exact butterfly counting is then applied on the sparsified graph to estimate the number of butterflies in the original graph. We present two such algorithms -- \edgSpr~ and \clr.

%and apply a scaling factor to estimate total number. We introduce edge-based and colorful algorithms in this class.
\item {\bf Local Sampling} algorithms, on the other hand, can work under limited access to the input graph. They randomly sample small subgraphs {\em local} to an element of the graph, and use them to compute an estimate. This is in contrast to sparsification, which needs a global view of the graph. We investigate sampling of subgraphs localized around a vertex (\verSam), edge (\edgSam), and a wedge (\wdgSam). Sampling algorithms are especially useful when there is a rate-limited API that provides random samples, such as the GNIP framework for Twitter~\cite{gnip} and the Graph API of Facebook~\cite{fbgraph}. {\it In the rest of this paper, when we say ``{\bf sampling algorithms}'', we mean local sampling algorithms.}
\end{itemize}

\item \textbf{Provable Guarantees:} We prove that the randomized algorithms yield estimates that are equal to the actual number of butterflies in the graph, in expectation. Through a careful analysis of their variance, we show that it is possible to reduce the estimation error to any desired level, through independent repetitions (sampling) or through an appropriate choice of parameters (sparsification).

\item \textbf{Experiments on real-world networks:} We present results of an evaluation of our algorithms on large real-world networks.
These results show that the algorithms can handle massive graphs with hundreds of millions of edges and trillions of butterflies.
Our most efficient sampling algorithm, which we call \edgSam+\fstEdg, gives estimates with a relative error less than $1$ percent within $5$ seconds, even for large graphs with trillions of butterflies. On such large graphs, (exact) algorithms from prior work took tens of thousands of seconds. %, and sometimes were unable to complete.
We observed a similar behavior with our best one-shot sparsification algorithm, \edgSpr, which typically yields estimates with error less than $1$ percent, within $4$ secs. 
%Among sparsification algorithms, \edgSpr performed bette
%Regarding the sampling, we show that a variant of edge sampling can give accurate estimations within a few seconds for all the networks in our dataset. 
%For the sparsification case, \edgSpr~{}performs way better than the colorful, and again give accurate estimations within a few seconds.

\end{itemize}

\vspace{-2ex}
\section{Preliminaries} 
\label{preliminaries}

We consider simple, unweighted, bipartite graphs, where there are no self-loops or multiple edges between vertices.
Let $G = (V, E)$ be a simple bipartite graph with $n=|V|$ vertices and $m=|E|$ edges.
Vertex set $V$ is partitioned into two sets $L$ and $R$ such that $V = L \cup R$ and $L \cap R = \emptyset$.
The edge set $E \subseteq L \times R$.
For $v \in V$, $\Gamma_v = \{u | (v, u) \in E\}$ denotes the set of vertices adjacent to $v$ (neighbors) and $d_v = \vert \Gamma_v \vert$ is the degree of $v$. 
$\Delta$ denotes the maximum degree of a vertex in the graph.
In addition, $\Gamma^2_v = \{w | (w, u) \in E \land w \ne v,~\forall~u \in \Gamma_v\}$ is the set of vertices that are \textit{exactly} in distance 2 from $v$, i.e., neighbors of the neighbors of $v$ (excluding $v$ itself). A wedge in $G$ is a path of length two.

A biclique is a complete bipartite subgraph, and is parameterized by the number of vertices in each partition; for integers $\alpha, \beta$,  an $\alpha \times \beta$ biclique in a bipartite graph is a complete subgraph with $\alpha$ vertices in $L$ and $\beta$ vertices in $R$.
A butterfly in $G$ is a $2\times 2$ biclique and consists of four vertices $\{a,b,x,y\} \subset V$ where 
$a,b \in L$ and $x,y \in R$ such that edges $(a,x)$, $(a,y)$, $(b,x)$, and $(b,y)$ all exist in $E$.
%It is same as the $(2,2)$-biclique, i.e., a complete bipartite graph of four vertices, two in each side.
Let $\bfly(G)$ denote the number of butterflies in $G$ (we use notation $\bfly$ when $G$ is clear from the context).
For vertex $v \in V$, let $\bfly_v$ denote the number of butterflies that contain $v$.
Similarly, for edge $e \in E$, let $\bfly_e$ denote the number of butterflies containing $e$ (See~\cref{fig:bfly-example}). 
Our goal is to compute $\bfly(G)$ for a graph $G$. We summarize our notation in~\cref{tab:notation}. 

When estimating the number of butterflies, we look for provable guarantees on the estimates computed using the following notion of randomized approximation. For parameters $\epsilon, \delta \in [0, 1]$, an $(\epsilon, \delta)-{}${}approximation of a number $Z$ is a random variable $\hat{Z}$ such that $\Pr[|\hat{Z} - Z| > \epsilon Z] \leq \delta$.

\begin{figure}[t!]
	\centering
	\captionsetup[subfigure]{captionskip=-4ex}
	\includegraphics[width=0.7\linewidth]{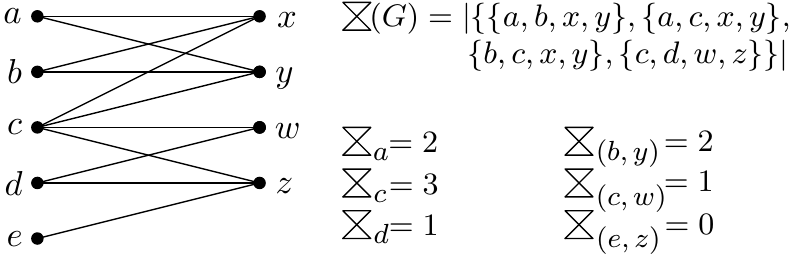}
	\vspace{0ex}
	\caption{There are 4 butterflies in the entire graph, and the number of per-vertex and per-edge butterflies are shown for some vertices/edges.}
	\label{fig:bfly-example}
	\vspace{-1ex} 
\end{figure}

\begin{table}[b!]
%\vspace{-1ex}
\footnotesize
\renewcommand{\tabcolsep}{2pt}
\linespread{0}\selectfont{}
\vspace{-1ex}
\linespread{1}\selectfont{}
\begin{tabular}{ | c | l |} \hline
$G=(V,E)$		& simple bipartite graph with vertices $V$ and edges $E$ \\ \hline
$V=L\cup R$	& vertex partitions $L$ and $R$ \\ \hline
$\Gamma_v$  &set of vertices adjacent to $v$, i.e., $\{u | (v, u) \in E\}$ \\ \hline
$d_v$	& degree of vertex $v$, i.e., $|\Gamma_v|$ \\ \hline
$n, m, \wdg$ & number of vertices ($|V|$), edges ($|E|$), and wedges ($\sum_{v \in V}{{d_v\choose {2}}}$) \\ \hline
$\Delta$ & maximum degree of a vertex in the graph \\ \hline
$\Gamma^2_v$ & distance-2 neighbors of $v$ (excluding itself)\\ \hline
% & i.e., $\{w | (w, u) \in E \land w \ne v,~\forall~u \in \Gamma_v\}$ \\ \hline
%$a,b$-biclique & Complete graph between $a$ vertices in $L$ and $b$ vertices in $R$\\ \hline
%butterfly & 2,2-biclique \\ \hline
$\bfly(G)$ (or $\bfly$) & number of butterflies in graph $G$\\ \hline
$\bfly_{v(e)}$  & number of butterflies that contain vertex $v$ (edge $e$) \\ \hline
%$p_{0v}$  & number of butterfly pairs that have no common vertices/edges \\ \hline
%$p_{1v}$  & number of butterfly pairs that have only one common vertex \\ \hline
%$p_{2v}$  & number of butterfly pairs that have two common vertices, no edges \\ \hline
%$p_{1e}$  & number of butterfly pairs that have one common edge \\ \hline
%$p_{1w}$  & number of butterfly pairs that have one common wedge \\ \hline
%$p_{V}$  & $p_{1v}+p_{2v}+p_{1e}+p_{1w}$ \\ \hline
%$p_{E}$  & $p_{1e}+p_{1w}$ \\ \hline
\end{tabular}
\caption{Notations}
\vspace{-5.0ex}
\label{tab:notation}
\end{table}

\paragraph{Networks and Experimental Setup.}
Since we frequently present evaluation results close to the algorithm descriptions, we give some details about the data used for evaluation.
We used massive real-world bipartite networks, selected from the publicly available KONECT network repository 
\footnote{\url{http://konect.uni-koblenz.de/}}. The graphs that we have used are summarized in~{}\cref{table:graphs}. 
We converted all graphs to simple, undirected graphs by removing edge directions (if graph was directed),  and by removing all self-loops, multiple edges, and vertices with degree zero.
We implemented all algorithms in C++,  compiled with Visual C++ 2015 compiler (Version 14.0), and report the runtimes on a machine equipped with a $3.50$ GHz Intel(R) Xeon(R) CPU E3-1241 v3 and $16.0$ GB memory. 

\begin{table}[t!]
	\small
	\centering
	\resizebox{\textwidth}{!}{
	\begin{tabular}{|c| @{\extracolsep{\fill}} |r|r|r|r|r|r|}
	\hline
	Bipartite graph     & $|L|$  & $|R|$ & $|E|$ & $\sum_{\ell\in{}L}{d_\ell^2}$ & $\sum_{r\in{}R}{d_r^2}$ & $\bfly$ \\ \hline
	$\prt{10^4,10}$-biclique & \num{10000} & \num{10} & \num{100000} & \num{1000000} & \num{1000000000} & \num{2249775000} \\ \hline
	\href{http://konect.uni-koblenz.de/networks/dbpedia-location}{DBPedia-Location}    & \num{172079} & \num{53407} & \num{293697} & \num{629983} & \num{245744643} & \num{3761594} \\ \hline
	\href{http://konect.uni-koblenz.de/networks/edit-frwiki}{Wikipedia (fr)}    & \num{288275} & \num{3992426} & \num{22090703} & \num{2186560482259} & \num{795544489} & \num{601291038864} \\ \hline
	\href{http://konect.uni-koblenz.de/networks/munmun_twitterex_ut}{\twt}   & \num{175214} & \num{530418} & \num{1890661} & \num{74184693} & \num{1943171143} & \num{206508691} \\ \hline
	\href{http://konect.uni-koblenz.de/networks/amazon-ratings}{Amazon}    & \num{2146057} & \num{1230915} & \num{5743258} & \num{828704126} & \num{437163474} & \num{35849304}  \\ \hline
	\href{http://konect.uni-koblenz.de/networks/livejournal-groupmemberships}{\jrn}     & \num{3201203} & \num{7489073} & \num{112307385} & \num{9565595799} & \num{5398290985343} & \num{3297158439527}  \\ \hline
	\href{http://konect.uni-koblenz.de/networks/edit-enwiki}{\wiki}     & \num{3819691} & \num{21416395} & \num{122075170} & \num{12576726351444} & \num{23256846234} & \num{2036443879822}  \\ \hline
	\href{http://konect.uni-koblenz.de/networks/delicious-ui}{\deli}     & \num{833081} & \num{33778221} & \num{101798957} & \num{85945241607} & \num{52787267117} & \num{56892252403}  \\ \hline
	\href{http://konect.uni-koblenz.de/networks/orkut-groupmemberships}{\ork}    & \num{2783196} & \num{8730857} & \num{327037487} & \num{156580049011} & \num{4900153043329} & \num{22131701213295}  \\ \hline
	\href{http://konect.uni-koblenz.de/networks/trackers-trackers}{\web}  & \num{27665730} & \num{12756244} & \num{140613762} & \num{1733678094524} & \num{211149951033914}  & \num{20067567209850}   \\ \hline
		\end{tabular}
}
%	\vspace*{-0cm}
	\caption{Bipartite network datasets. $L$ and $R$ are vertex partitions, $E$ is the edge set. The sum of degree squares for $L$ and $R$, and the number of butterflies are shown.
	\label{table:graphs}} 
	\vspace{-2ex}
\end{table}

\vspace{-2ex}
\section{Related Work}\label{sec:related}
%\noindent We summarize the state-of-the-art studies on bipartite graph motifs, including butterfly counting.
% and the butterfly counting in particular.\\

\noindent \textbf{Bipartite graph motifs:} Modeling the smallest unit of cohesion enables a principled way to analyze networks.
While the literature is quite rich with the studies on counting triangles and small cliques in unipartite graphs~\cite{SePiKo14, Jha15, Pinar17, Ahmed16, PTTW13, PTT13a, Bressan17, Jain17}, these works are not applicable to bipartite networks.
To the best of our knowledge, Borgatti and Everett~\cite{Borgatti97} are the first to consider cohesive structures in bipartite networks to analyze social networks. They proposed to use the $3\times 3$ biclique as the smallest cohesive structure, motivated by the fact that a triangle in a unipartite graph has three vertices, and the same should be considered for both vertex sets of the bipartite graph. Opsahl also proposed a similar approach to define the clustering coefficient in affiliation networks~\cite{Opsahl13}.
Robins and Alexander argued that the smallest structure with multiple vertices on both vertex sets is a better model for measuring cohesion in bipartite networks~\cite{Robins04}, as also discussed in a later work~\cite{Lind05}.
They used the ratio of the number of $2\times 2$ bicliques (butterflies) to the number of $3$-paths (a path of three edges) to define the clustering coefficient in bipartite graphs. The butterfly is also adopted in a recent work by Aksoy et al.~\cite{Aksoy16} to generate bipartite graphs with community structure.
%Another field that the butterflies are used is the study of near-optimal decoding algorithms.
Butterfly counting is also applicable to the study of graphical codes (Halford and Chugg~\cite{Halford06}). The numbers of cycles of length $g$, $g+2$, and $g+4$ in bipartite graphs, where $g$ is the girth\footnote{length of the shortest cycle in a graph}, characterize the decoding complexity -- note that the butterfly is the simplest non-trivial cycle in a bipartite graph.
% ``Good graphical code models are those that imply near-optimal decoding algorithms with practically realizable complexity''.
%Note that a butterfly can be thought of as the simplest cycle in a bipartite network, since a 4-cycle is the same as the butterfly, and there is no 3-cycle in a bipartite graph.
%We believe that our fast sampling algorithms for counting butterflies will be beneficial for more efficient decoding algorithms.

%\erdem{This paragraph is included to say that we tried the existing techniques and find this and that.}\\
\textbf{Random Sampling for Motif Counting:} There has been much interest in recent years for counting motifs via random sampling, mostly focused on unipartite graphs and triangles. Works on triangle counting include edge sampling~\cite{Babis09}, subsequently improved by colorful sampling~\cite{colorful-triangles}, on vertex and edge sampling~\cite{Wu16, Itai78}, wedge sampling~\cite{SePiKo14}, a hybrid scheme that considers the edge and wedge sampling~\cite{Turkoglu17}, neighborhood sampling~\cite{PTTW13}, and a recent space-efficient algorithm~\cite{DeStefani16} based on reservoir-sampling. To the best of our knowledge, random sampling for butterfly counting in bipartite networks has not been studied in the past.
%Pagh and Tsourakakis introduced \textit{colorful sampling}~\cite{colorful-triangles} that chooses the edges interdependently to improve over previous results~\cite{Babis09}.
%The idea is to randomly color the vertices in the graph and choose the edges which have same colors on both ends.

\textbf{Butterfly Counting:} The closest work to ours is by Wang et al.~\cite{Wang14}, who presented exact algorithms for butterfly (rectangle) counting. Their algorithms outperform generic matrix-multiplication based methods for counting cycles in a graph~\cite{Alon97}. We present an improved algorithm for exact butterfly counting, and then present more efficient randomized algorithms for approximate butterfly counting.

%\input{setup}
%-------------------------------------------
\section{Exact Butterfly Counting}
\label{sec:exact}
%-------------------------------------------

\noindent We first present the basic equation for the number of butterflies in a bipartite graph $G$ and the base (state-of-the-art) algorithm by Wang et al.~\cite{Wang14} that implements the equation.

%----------------
\begin{lem}
\vspace{-2ex}
\label{lem:exactBFC}
For a bipartite graph $G = (V = (L \cup R), E)$, \\
%\begin{equation} \label{eq:1}
$(1)~\bfly(G) = \frac{1}{2}\sum_{v \in L} {\bfly_v}$. \\
%\end{equation}
%\begin{equation} \label{eq:2}
$(2)~\bfly_v = \sum_{u \in \Gamma_v} \sum_{\tiny \textnormal{(unique)}~w \in \Gamma_u \setminus {v}}{|\Gamma_v \cap \Gamma_w| \choose 2} = \sum_{w \in \Gamma^2_{v}}{|\Gamma_v \cap \Gamma_w| \choose 2}$.\\
%\end{equation}
%\begin{equation} \label{eq:3}
$(3)~\bfly(G) = \frac{1}{2}\sum_{v \in L} \sum_{w \in \Gamma^2_{v}}{|\Gamma_v \cap \Gamma_w| \choose 2}$.
%\end{equation}
\end{lem}
%----------------

\begin{proof}
As each butterfly has exactly two vertices in the set $L$, Equation~(1) holds.
For a vertex $v$ in $L$, each butterfly it participates has one other vertex $w \in L~(w \ne v$) and two vertices $u, x \in R$.
By definition, $w \in \Gamma^2_v$.
In order to find the number of $u, x$ pairs that $v$ and $w$ form a butterfly, we compute the intersection of the neighbor sets of $v$ and $w$.
%Note that, size of the intersection should be at least two in order to have a butterfly.
Number of the pairs in the intersection is defined by $|\Gamma_v \cup \Gamma_w| \choose 2$, as in Equation (2).
Replacing $\bfly_v$ in Equation~(1) by the right-most hand side of Equation~(2), we obtain Equation (3).
\end{proof}

\newcommand\mycommfont[1]{\ttfamily\textcolor{black}{#1}}
\SetCommentSty{mycommfont}
\begin{algorithm}[!t]
\small
\DontPrintSemicolon
\SetKwInOut{Input}{Input}
\SetKwInOut{Res}{Output}
 \Input{Graph $G = (V = (L, R), E)$}
 \Res{$\bfly(G)$}

    $\mathcal{A} \gets L, \bfly \gets 0$\textcolor{magenta}{\;\label{ln:1}
    \If {$\sum_{u \in L}{(d_u)^2} < \sum_{v \in R}{(d_v)^2}$} {\label{ln:2}
        $\mathcal{A} \gets R$\;\label{ln:3}
    }}
    \For {$v\in \mathcal{A}$} {\label{ln:4}
	$\mathcal{C} \leftarrow hash map$~\tcp*{initialized with zero}\label{ln:5}
        \For{$u\in \Gamma_v$} {\label{ln:6}
            \For{$w\in \Gamma_u~$\textcolor{magenta}{:~$w \prec v$}} {\label{ln:7}
            	$\mathcal{C}[w] \gets \mathcal{C}[w] + 1$~\tcp*{dist-2 multiplicities}\label{ln:8}
            }
        }
        \For{$w\in \mathcal{C} : \mathcal{C}[w] > 0$} {\label{ln:9}
        	$\bfly \gets \bfly + {\mathcal{C}[w] \choose 2}$\label{ln:10}\;
        }
    }
    \Return $\bfly/2$ \textcolor{magenta}{($\bfly$)}\label{ln:11}\;
\caption{\exact~(V, E): Exact Butterfly Counting
 \label{algo:exactBFC}}
\end{algorithm}
%------------------

\noindent An efficient implementation of Equation~(3) is given in~\cref{algo:exactBFC} \exact~(ignore the colored lines).
Instead of performing set intersection at each step, we count and store the number of paths from a vertex $v \in L$ to each of its distance-2 neighbor $w \in L$ by using a hash map $\mathcal{C}$ (in \cref{ln:5} to~\cref{ln:8}).
%Multiplicity of a vertex $w \in \mathcal{C}$ gives the size of intersection between $\Gamma_v$ and $\Gamma_w$.
%Number of butterflies that contain $v$ and $w$ is equal to the number of neighbor pairs they have in common, which is $\mathcal{C}[w] \choose 2$ for $w \in L$.
%Thus we sum those in lines~\ref{ln:9} and~\ref{ln:10}.
%Note that the total sum after the main loop (lines~\ref{ln:4} to~\ref{ln:10}) counts each butterfly twice and we halve it to get the correct result in line~\ref{ln:11}.
Additional space complexity of \exact~is $\bigO(|V|)$, which is used by the hash map $\mathcal{C}$.
Computational complexity of \exact~(without colored lines) is given by the following.

%-------------
\begin{lem} 
\label{lem:cost-exact}
On input graph $G= (V = (L \cup R), E)$, time complexity of \textnormal{\exact}~(without colored lines) is $\bigO(\sum_{u \in R}{d_u^2})$.
\end{lem}
%-------------

%-------------
\begin{proof}
For each $v \in L$, we find a distance-2 neighbor vertex $w \in L$ such that there exists a $u \in R$ where $(v,u) \in E$ and $(w,u) \in E$.
In \cref{ln:6} to~\cref{ln:8}, the nested \texttt{for} loop performs $\bigO(1)$ computation for each tuple $(v,u,w)$ where $v \in L, u \in \Gamma_v,$ and $w \in \Gamma_u$. \textit{The number of such triples is exactly the number of paths in the graph of length two, with the midpoint ($u$) in $R$}, which is equal to $\sum_{u \in R}{d_u \choose 2} = \bigO(\sum_{u \in R} (d_u)^2)$.
\end{proof}
%-------------

We have two observations about Equation~(3) and~\cref{algo:exactBFC}.
First, the intersection operation need not be performed for each ordered pair $v, w$ in $L$, but instead can be performed for all (ordered) pairs $w \prec v$, thus preventing double counting of a butterfly.
%, thus each butterfly is counted twice which is fixed at the end by halving the total sum.
%To prevent the redundant computation, we put an ordering constraint for each $v, w$ pair in $L$ so that each butterfly is counted only once and the total work done is reduced to half.
%This is done by the $w \prec v$ constraint (in blue) at line~\ref{ln:7}.
%Therefore we do not have to halve the sum anymore in line~\ref{ln:11}.
%\begin{equation} \label{eq:ordering}
%\bfly(G) = \sum_{v \in L} \sum_{\substack {w \in \Gamma^2_{v} \\ w\prec v}}{|\Gamma_v \cap \Gamma_w| \choose 2}.
%\end{equation} 

\noindent Second, instead of iterating over the vertex set $L$ in~Equation~(3), we can also use other vertex set, $R$;\\
%\begin{equation} 
$(4)~\bfly(G) = \frac{1}{2}\sum_{u \in R} \sum_{x \in \Gamma^2_{u}}{|\Gamma_u \cap \Gamma_x| \choose 2}$
%\end{equation}

Clearly,~Equations~(3) and~(4) give the same answer overall, but their computational costs can be significantly different.
We show (\cref{lem:cost-exact}) that the runtime of the algorithm is $\bigO(\sum_{u \in R} (d_u)^2)$ if we use $\mathcal{A} = L$.
Based on this analysis, we propose to use an $\bigO(n)$ time pre-computation step that compares the sum of degrees in $L$ and $R$ to choose the cheaper option. If $\sum_{v \in L} (d_v)^2 < \sum_{u \in R} (d_u)^2$, then we choose the right side, $R$; otherwise we choose the left side, $L$.
This is done in the \cref{ln:2,ln:3} (in red).

\exact~algorithm (including colored lines) has the complexity of $\bigO(min (\sum_{u \in R}{(d_u)^2}, \sum_{v \in L}{(d_v)^2})$ and improves upon the time complexity of the algorithm due to Wang et al.~\cite{Wang14}, which is $\bigO\left(\sum_{u \in R}{(d_u)^2}\right)$. The main difference is that their algorithm always starts from the left vertex set $L$, while we choose the cheaper option depending on the sum of degrees in each side.

% A comparison of the runtime of the two algorithms is shown in \cref{figure:exactVSwang}. 
% For instance, in the {\em \twt} graph, \exact{}~takes about $0.3$ seconds, while the algorithm of~\cite{Wang14} takes more than 10 seconds. On the {\em \ork} graph, \exact~{}took 432 sec -- $35.3$x faster than the algorithm of~\cite{Wang14} which took $15281$ sec.

%----------------------------------------
\vspace{-2ex}
\subsection{Performance of Exact Butterfly Counting}
%----------------------------------------

%Another aspect of our algorithm is that it maintains the value of ${{\Gamma_{\ell} \cap \Gamma_{\ell'}} \choose 2}$ incrementally (see Line 9 of \cref{algo:exactBFC} where this size is maintained in the variable {\texttt{inCommon}}) avoiding an explicit computation of this quantity, which will need us to re-evaluate the size of a set intersection.

We compare the runtime of our algorithm (\exact) with Wang et al.~\cite{Wang14} (\wangExact).
From our theoretical analysis in \cref{sec:exact}, our algorithm is expected to be faster than \wangExact. 
\cref{figure:exactVSwang} shows a comparison of the runtimes of the two algorithms. We note the following points.
(1)~{}\exact~{}is always faster than \wangExact. This also shows that the $O(n)$ pre-processing step in \exact~{}to choose
which side to proceed from is effective. 
(2)~{}In many cases, \exact~{}achieves significant speedup when compared with \wangExact.
This is especially true in cases where $\sum_{\ell \in L}{d_\ell^2}$ and $\sum_{r \in R}{d_r^2}$ differ significantly. 
In particular, \exact~{}is 700 times faster than \wangExact~{}on \jrn~{}network and $35$ times faster on \ork. For the \web~{}network, \wangExact~ did not complete in $40,000$ seconds while \exact~{}completed in $\approx 9,000$ secs. 
%We also generated a (${10}^4, 10$) biclique to highlight the importance of choosing the better side; for this graph, \exact~{}is $60x$ faster. 
%These results show that our exact algorithm is able to handle very large graphs with trillions of bicliques.
%\erdem{or you can put an upward arrow on top, like fig. 6 in http://sariyuce.com/papers/vldb17.pdf}

\begin{algorithm}[!t]
    \DontPrintSemicolon
    \KwIn{A vertex $v \in V$ in $G = (V = (L \cup R), E)$}
    \KwOut{$\bfly_v$, number of butterflies in $G$ that contain $v$}
    $\bfly_v \gets 0$, $\mathcal{C} \leftarrow hash map$~\tcp*{initialized with zero}
    \For{$u\in \Gamma_v$} {
    	\lFor{$w\in \Gamma_u~$} {
	        \IfThen{$w \ne v$}{$\mathcal{C}[w] \gets\mathcal{C}[w] + 1$}
%        		$\mathcal{C}[w] \gets \mathcal{C}[w] + 1$
        }
    }
    \lFor{$w\in \mathcal{C} : \mathcal{C}[w] > 0$} {
    	$\bfly_v \gets \bfly_v + {\mathcal{C}[w] \choose 2}$
    }
    \Return $\bfly_v$\;
    \caption{\pervertex~(v, G): Per Vertex Butterfly Counting}
    \label{algo:local-verBFC}

\end{algorithm}

\begin{figure}[!b]
%\vspace{-2ex}
	\centering
%	newResults/histogram/exact-vs-IMRC-wit-arrow
	%\resizebox{0.8\linewidth}{!}{originPro/histogram}
	
	\includegraphics[width=\linewidth]{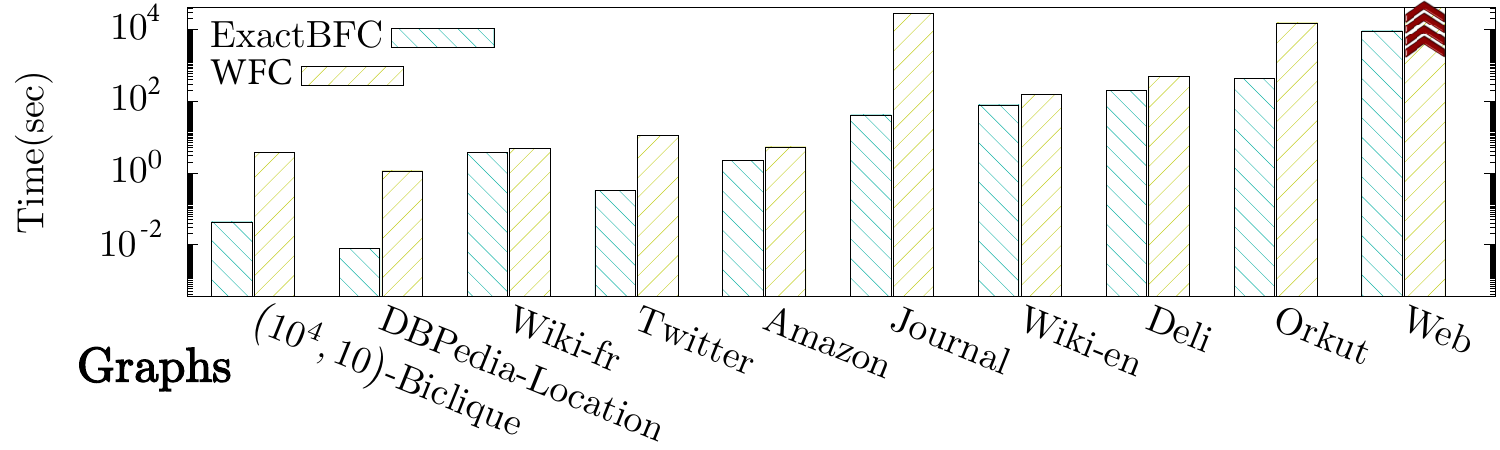}
	\vspace{-20px}
	\caption{Runtimes for Exact Butterfly Counting, showing speedups up to 3 orders of magnitude for \exact{}~over \wangExact. For the \web~graph, \wangExact~{}did not finish in \num{40000} sec %.\erdem{y-axis is wrong, \exact{} takes less than 8K secs on web? Please also extend y-axis until $10^4$}
		\label{figure:exactVSwang}}
\vspace{-4ex}
\end{figure}

\vspace{-2ex}
\subsection{Local Butterfly Counting}
We present two algorithms for local butterfly counting, \pervertex~(\cref{algo:local-verBFC}) for counting the number of butterflies $\bfly_v$ that contain a given vertex $v$, and \peredge~(\cref{algo:local-edgBFC}) for counting the number of butterflies $\bfly_e$ that contain an edge $e$. Both algorithms employ procedures similar to the inner loop (\cref{ln:5} to~\cref{ln:10}) of \exact. We will use \pervertex~and \peredge~ as building blocks in our sampling algorithms (\cref{sec:sampling}).

%--------------------

%--------------------

\begin{algorithm}[t!]
    \DontPrintSemicolon
    \KwIn{An edge $e = (u, v) \in E$ in $G = (V, E)$}
    \KwOut{$\bfly_e$, number of butterflies in $G$ that contain $e$}
    $\bfly_e \gets 0$, $\mathcal{C} \leftarrow hash map$~\tcp*{initialized with zero}
    \For{$w\in \Gamma_u$}{
    	\lFor{$x\in \Gamma_w$}{
	    \IfThen{$x \in \Gamma_v \setminus \{u\}$}{$\mathcal{C}[x] \gets\mathcal{C}[x] + 1$}}}
     \lFor{$x\in \mathcal{C}$} { % : \mathcal{C}[x] > 0
    	$\bfly_e \gets \bfly_e + {\mathcal{C}[x] \choose 2}$
    }
    \Return $\bfly_e$\;
    \caption{\peredge~(e, G): Per Edge Butterfly Counting}
    \label{algo:local-edgBFC}
\end{algorithm}

\begin{lem}
\vspace{-2ex}
Time complexity for \textnormal{\pervertex $(v, G)$} is $\bigO(|\Gamma^2(v)|)=\bigO(\Delta~d_v)$ where $\Delta$ is the maximum degree in $G$, and for \textnormal{\peredge$(e$=$(u,v), G)$} is $\bigO(|\Gamma^2(u)| + |\Gamma_v|)=\bigO(\Delta~d_u + d_v)$ where $d_u < d_v$ w.l.o.g..
\vspace{-2ex}
\end{lem}

\begin{proof}
\textnormal{\pervertex} iterates through each vertex in $\Gamma^2(v)$ spending $\bigO(1)$ time per iteration. 
Since $|\Gamma^2(v)|$ is bounded by $\Delta~d_v$, the time complexity follows. 
For \textnormal{\peredge}, the only extra work is to check whether an $x \in \Gamma^2_u$ also exists in $\Gamma_v$ and it can be done by a hash map which requires $d_v$ preprocessing time.
\end{proof}

%------------------------------------------------------------------------------
%\section{Approximate Butterfly Counting Using Sampling}
\section{Approximation by Local Sampling}
\label{sec:sampling}
%WE SURVEY POPULAR TECHNIQUES USED FOR APPROXIMATE MOTIF COUNTING IN LARGE NETWORKS FOR BIPARTITE GRAPHS. WE ANALYZE THEM AS WELL.
%------------------------------------------------------------------------------

In this section, we present approaches to approximating $\bfly(G)$ using random sampling. The intuition behind sampling is to examine a randomly sampled subgraph of $G$ and compute the number of butterflies in the subgraph to derive an estimate of $\bfly(G)$. Since the subgraph is typically much smaller than $G$, it is less expensive to perform an exact computation. The size of the chosen subgraph, the cost of computing on it, and the accuracy of the estimate vary according to the method by which we choose the random subgraph. \textit{This sampling process can be repeated multiple times, and averaged, in order to get a better accuracy.} 

% and is inexpensive for exact computation.
%Depending on how often a butterfly is observed in the sampled subgraph, we estimate the total number of butterflies in the entire graph. 
%There are three popular techniques that have been shown to be simple and effective: 
%Here we adapt those three approaches for approximate butterfly counting in a bipartite network and analyze their variances.
%We obtain an average number of butterflies per vertex upon examining multiple samples, and extrapolate this number to get an approximation.
%Literature is rich with various sampling approaches to approximately count the graph motifs such as triangles and small cliques in unipartite graphs.

We consider three natural sampling methods: vertex sampling (\verSam), edge sampling (\edgSam), and wedge sampling (\wdgSam). In \verSam, the subgraph is chosen by first choosing a vertex uniformly at random, followed by the induced subgraph on the distance-2 neighborhood of the vertex. In \edgSam, a random edge is chosen, followed by the induced subgraph on the union of the immediate neighborhoods of the two endpoints of the edge. In \wdgSam, a random wedge (path of length two) is chosen, followed by the induced subgraph on the  intersection of the immediate neighborhoods of the two endpoints of the wedge. While the methods themselves are simple, the analysis of their accuracy involves having to deal carefully with the interactions of different butterflies being sampled together. 

%------------------------------------------------------------------------------------------------------------------------------------------------------------------------------------------------
\begin{figure*}[ht!]
\centering
%\resizebox{1\linewidth}{!}{
\subfloat[][\scriptsize{Type $0v$: No vertex in common}]{\includegraphics[width=0.19\textwidth]{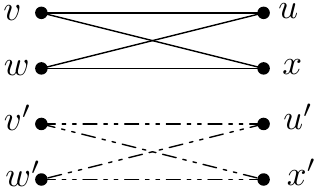}\label{verSamA}}~{}~{}
%\hspace{3ex}
\subfloat[][\scriptsize{Type $1v$: One vertex in common}]{\includegraphics[width=0.19\textwidth]{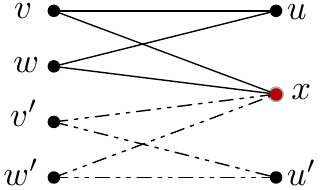}\label{verSamB}}~{}~{}
%\hspace{3ex}
\subfloat[][\scriptsize{Type $2v$: Two vertices in common but no edge}]{\includegraphics[width=0.19\textwidth]{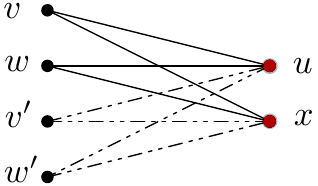}\label{verSamC}}~{}~{}
%\hspace{3ex}
\subfloat[][\scriptsize{Type $1e$: One edge in common}]{\includegraphics[width=0.19\textwidth]{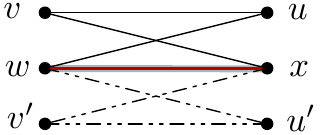}\label{verSamD}}~{}~{}
%\hspace{3ex}
\subfloat[][\scriptsize{Type $1w$: One wedge in common}]{\includegraphics[width=0.19\textwidth]{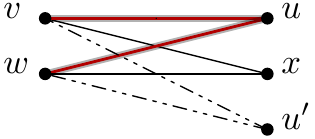}\label{verSamE}}
%}
\vspace{-3.3ex}
\caption{A pair of butterflies in $G$ can be of one of the above five types.}%\erdem{Left vertices should be v, w and rights are u, x. - fixed}}
\vspace{-2ex}
\end{figure*}
%------------------------------------------------------------------------------------------------------------------------------------------------------------------------------------------------
%------------------------------------------------------------------------------------------------------------------------------------------------------------------------------------------------
%------------------------------------------------------------------------------------------------------------------------------------------------------------------------------------------------
%------------------------------------------------------------------------------------------------------------------------------------------------------------------------------------------------
%------------------------------------------------------------------------------------------------------------------------------------------------------------------------------------------------
\subsection{Vertex Sampling (Algorithm~\verSam)}
%---------------------------------------
The idea in \verSam~{}is to sample a random vertex $v$ and count the number of butterflies that contain $v$ -- this is accomplished by counting the number of butterflies in the induced subgraph consisting of the distance-2 neighborhood of $v$ in the graph. We show that the algorithm, described in \cref{algo:versamp}, yields an unbiased estimate of $\bfly(G)$, and also analyze the variance of the estimate. The variance is reduced by taking the mean of multiple independent runs of the estimator.
%------------------------------------------------
\begin{algorithm}[t]
    \DontPrintSemicolon
    \KwIn{A bipartite graph $G = (V, E)$}
    \KwOut{An estimate of $\bfly(G)$}
    Choose a vertex $v$ from $V$ uniformly at random.\;
    $\bfly_v \gets \pervertex~(v, G)$~\tcp*{\cref{algo:local-verBFC}}
%    Compute $\bfly_v$ using \cref{algo:local-verBFC} for local (per-vertex) butterfly counting \;
    \Return $\bfly_v \cdot n/4$
    \caption{\verSam~(single iteration)
    \label{algo:versamp}}
%    \vspace{-1ex}
\end{algorithm}

%------------------------------------------------
%------------------------------------------------
%------------------------------------------------
Let $Y_V$ denote the return value of \cref{algo:versamp}. Let $p_V$ denote the number of pairs of butterflies in $G$ that share a single vertex.
%------------------
\begin{lem}
\vspace{-1ex}
\label{lem:versamp-exp}
\label{lem:versamp-var}
\label{lem:versamp}
$\expec{Y_V} = \bfly$, and $\var{Y_V} \le \frac{n}{4}(\bfly + p_V)$
\vspace{-2ex}
\end{lem}
%------------------------------------------------
%------------------------------------------------
\begin{proof}
Consider that the butterflies in $G$ are numbered from $1$ to $\bfly$.
Let $X = \bfly_v$, the number of butterflies that contain the vertex $v$, which is sampled uniformly.
For $i = 1, \ldots, \bfly$, let $X_i$ be an indicator random variable equal to $1$ if the $i^{\text{th}}$ butterfly includes the vertex $v$.
We have $X = \sum_{i=1}^{\bfly} X_i$. Since each butterfly has four vertices, $\expec{X_i} = \prob{X_i=1} = 4/n$.
Thus $\expec{X} = \sum_{i=1}^{\bfly} \expec{X_i} = \sum_{i=1}^{\bfly} \prob{X_i=1} = \frac{4\bfly}{n}$.
Since $Y_V = X \cdot \frac{n}{4}$, we have $\expec{Y_V} = \bfly$.
%\end{proof}
%------------------------------------------------

For the variance of $Y_V$, we consider the joint probabilities of different butterflies being sampled together. The set of all pairs of butterflies are partitioned into different types as follows. This partitioning will help not only with analyzing this algorithm, but also in subsequent sampling algorithms. A pair of butterflies is said to be of type:
\begin{itemize}
\vspace{-1ex}
\setlength\itemsep{0pt}
\item $0v$ if they share zero vertices (\cref{verSamA})
\item $1v$ if they share one vertex (\cref{verSamB})
\item $2v$ if they share two vertices but no edge (\cref{verSamC})
\item $1e$ if they share two vertices and exactly one edge (\cref{verSamD})
\item $1w$ if they share three vertices and two edges i.e. share a wedge (\cref{verSamE})
\end{itemize}
It can be verified that every pair of distinct butterflies must be one of the above types $\{0v,1v,2v,1e,1w\}$, and other combinations such as three vertices and more than two edges are not possible. For each type $t \in \{0v,1v,2v,1e,1w\}$, let $p_t$ denote the number of pairs of butterflies of that type. Thus $p_{0v} + p_{1v} + p_{2v} + p_{1e} + p_{1w} = {\bfly\choose{2}}$.  Let $p_V = p_{1v} + p_{2v} + p_{1e} + p_{1w}$ be the number of pairs of butterflies that share at least one vertex.
%------------------------------------------------
%\begin{lem}\label{lem:versamp-var}
%\end{lem}
%------------------------------------------------
%------------------------------------------------
%\begin{proof}
%\erdem{This can be omitted, not sure though}
%\erdem{Just keep one proof that uses Fig.3 in the paper. Put the rest in appendix or remove.}
\begin{small}
  \begin{equation*}
  \begin{aligned}
    \var{Y_V} &= \var{\frac{n}{4}\sum_{i=1}^{\bfly}{X_i}} = \frac{n^2}{16} \var{\sum_{i=1}^{\bfly}{X_i}} \\ %\frac{n^2}{4}\sum_{i=1}^{\bfly}\sum_{j=1}^{\bfly}{\Cov(i, j)} \\
     &= \frac{n^2}{16}\left[\sum_{i=1}^{\bfly}{\var{X_i}} + \sum_{i\neq{}j}{\Cov\lbr X_i, X_j\rbr}\right] \\
     %&= \frac{n^2}{16}\left[\sum_{i = 1}^{\bfly}{\lbr\expec{X_i^2} - \expec{X_i}^2\rbr} + \sum_{i\neq{}j}{\Cov\lbr X_i, X_j\rbr}\right] \\
     &= \frac{n^2}{16}\left[\bfly \lbr\frac{4}{n} - \frac{16}{n^2}\rbr + \sum_{i\neq{}j}{\lbr\expec{X_iX_j} - \expec{X_i}\expec{X_j}\rbr}\right]
  \end{aligned}
  \end{equation*}
\end{small}
%In order to compute $\sum_{i\neq{}j}{\Cov\lbr X_i, X_j\rbr}$, we examine all kinds of butterfly pairs:
Consider the different types of butterfly pairs $(i,j)$:
\begin{itemize}[leftmargin=0pt]
\setlength\itemsep{0pt}
\setlength{\parskip}{0pt}
\setlength{\parsep}{0pt}
\item
Type $0v$, there is zero probability of $i$ and $j$ being counted within $\bfly_v$, hence $\expec{X_i X_j} =  0$, $\Cov\lbr X_i, X_j \rbr = -{16}/{n^2}$
\item
Type $1v$, $\E[X_iX_j] = \Pr[X_i=1] \Pr[X_j=1 \vert X_i=1] = ({4}/{n})({1}/{4}) = {1}/{n}$. $\Cov\lbr X_i,X_j \rbr= {1}/{n} - {16}/{n^2}$.
\item
Type $2v$, $\E[X_iX_j] = ({4}/{n})({1}/{2}) = {2}/{n}$.  $\Cov\lbr X_i,X_j \rbr= {2}/{n} - {16}/{n^2}$%$\E[X_iX_j] = \Pr[X_i=1]\Pr[X_j=1|X_i=1] = ({4}/{n})({1}/{2}) = {2}/{n}$.  $\Cov\lbr X_i,X_j \rbr= {2}/{n} - {16}/{n^2}$.
\item
Type $1e$, $\E[X_iX_j] = ({4}/{n})({1}/{2}) = {2}/{n}$. $\Cov\lbr X_i,X_j \rbr= {2}/{n} - {16}/{n^2}$
%$\E[X_iX_j] = \Pr[X_i=1]\Pr[X_j=1|X_i=1] = ({4}/{n})({1}/{2}) = {2}/{n}$. $\Cov\lbr X_i,X_j \rbr= {2}/{n} - {16}/{n^2}$.
\item
Type $1w$, $\E[X_iX_j] = ({4}/{n})({3}/{4}) = {3}/{n}$. $\Cov\lbr X_i,X_j \rbr= {3}/{n} - {16}/{n^2}$.
%$\E[X_iX_j] = \Pr[X_i=1]\Pr[X_j=1|X_i=1] = ({4}/{n})({3}/{4}) = {3}/{n}$. $\Cov\lbr X_i,X_j \rbr= {3}/{n} - {16}/{n^2}$.
\end{itemize}
By adding up the different contributions, we arrive:
{\small \[
        \Var[Y_V] \le \frac{n}{4}\lbr\bfly + p_{1v} + p_{2v} + p_{1e} + p_{1w}\rbr - {\bfly\choose{2}} \leq \frac{n}{4}\lbr\bfly + p_V\rbr
\]}
\myremove{
\begin{small}
      \begin{equation*}
      \begin{aligned}
        \Var[Y_V]  %&=  \frac{n^2}{16}\bigg[\bfly \lbr\frac{4}{n} - \frac{16}{n^2}\rbr - p_{0v}\frac{16}{n^2} + p_{1v}(\frac{1}{n}  - \frac{16}{n^2})  \\
                      %&  + p_{2v}\lbr\frac{2}{n} - \frac{16}{n^2}\rbr + p_{1e}\lbr\frac{2}{n} - \frac{16}{n^2}\rbr + p_{1w}\lbr\frac{3}{n} - \frac{16}{n^2}\rbr\bigg]\\
         %& \leq \frac{n^2}{16}\left[\frac{4}{n}\lbr\bfly + p_{1v} + p_{2v} + p_{1e} + p_{1w}\rbr - {\bfly\choose{2}}\frac{16}{n^2}\right]\\
         & \le \frac{n}{4}\lbr\bfly + p_{1v} + p_{2v} + p_{1e} + p_{1w}\rbr - {\bfly\choose{2}} \leq \frac{n}{4}\lbr\bfly + p_V\rbr
      \end{aligned}
      \end{equation*}
\end{small}
}
\vspace{-1ex}
\end{proof}
%------------------------------------------------
Let $Z$ be the  average of $\alpha = \frac{8n}{\epsilon^2\bfly}\left(1 + \frac{p_V}{\bfly} \right)$ independent instances of $Y_V$. Using $\Var[Z] = \Var[Y_V] / \alpha$
and Chebyshev's inequality:
  \begin{equation*}
	  \begin{aligned}
	     \prob{|Z - \bfly| \geq \epsilon\bfly}  &\leq \frac{\Var[Z]}{\epsilon^2{}\bfly^2} = \frac{\Var[Y_V]}{\alpha\epsilon^2{}\bfly^2}
	     \leq \frac{n(\bfly + p_V)}{4\alpha\epsilon^2{}\bfly^2} = \frac{1}{32}
	  \end{aligned}
  \end{equation*}
We can turn the above estimator into an $(\epsilon,\delta)$ estimator by taking the median of $\bigO(\log(1/\delta))$ estimators, using standard methods. 
%---------------
\begin{lem}
There is an algorithm that uses $\tilde{\bigO}\left(\frac{n}{\bfly}\left(1 + \frac{p_V}{\bfly} \right) \right)$
iterations~\footnote{We use the notation $\tilde{\bigO}(f)$ to suppress the factor $\frac{\log(1/\delta)}{\epsilon^2}$,
i.e. mean $\bigO\left(f \cdot \frac{\log(1/\delta)}{\epsilon^2}\right)$}
 of \textnormal{\verSam} (\cref{algo:versamp}) and yields an $(\epsilon,\delta)$-estimator of $\bfly(G)$ using expected time $\tilde{\bigO}\left(\frac{m\Delta}{\bfly}\left(1 + \frac{p_V}{\bfly} \right) \right)$. Expected additional space is $\bigO\left(\frac{m\Delta}{n}\right)$.
\end{lem}
%------------
\begin{proof}
An iteration of \textnormal{\verSam} samples a vertex $v$ and calls \textnormal{\pervertex} (\cref{algo:local-verBFC}) for local butterfly counting once, which takes $\bigO(|\Gamma^2_v|)$ time. Hence, the expected runtime of an iteration is $\bigO\left(\expec{|\Gamma^2_v|}\right)$, where the expectation is taken over a uniform random choice of a vertex. We note that $|\Gamma^2_v| \le d_v~\Delta$ where $d_v$ is $v$'s degree and $\Delta$ is the maximum degree in the graph. Thus $\expec{|\Gamma^2_v|} \le \sum_{v \in V} \frac{1}{n} d_v \Delta = \frac{\Delta}{n} \sum_{v \in V} d_v = \frac{m\Delta}{n}$.
The space of \verSam~{}is same with \textnormal{\pervertex} (\cref{algo:local-verBFC}); $\bigO(|\Gamma_v^2|)$ for handling vertex $v$. The expected value is $\bigO(\frac{m\Delta}{n})$.%Note that the space cost is in addition to the cost of storing the graph itself, which is only needed as read-only access.
\end{proof}

\begin{algorithm}[t!]
    \DontPrintSemicolon
    \SetKwInOut{Input}{Input}
    \Input{A bipartite graph $G = (V, E)$}
    \KwOut{An estimate of $\bfly(G)$}
    Choose an edge $e$ from $E$ uniformly at random.\;
    $\bfly_e \gets \peredge~(e, G)$~\tcp*{\cref{algo:local-edgBFC}}\label{ln:esmp}
%    Compute $\bfly_e$ using \cref{algo:local-edgBFC}\;
    \Return $\bfly_e  \cdot m/4$
    \caption{\edgSam~(single iteration)
    \label{algo:edgsamp}}
\end{algorithm}

%------------------------------------------------------------------------------------------------------------------------------------------------------------------------------------------------
%------------------------------------------------------------------------------------------------------------------------------------------------------------------------------------------------
%------------------------------------------------------------------------------------------------------------------------------------------------------------------------------------------------
%------------------------------------------------------------------------------------------------------------------------------------------------------------------------------------------------
%------------------------------------------------------------------------------------------------------------------------------------------------------------------------------------------------
%------------------------------------------------------------------------------------------------------------------------------------------------------------------------------------------------
\subsection{Edge Sampling (Algorithm~\edgSam)}
In this algorithm, the idea is to sample a random edge and count the number of butterflies that contain this edge, using \peredge~(\cref{algo:local-edgBFC}).
We present \edgSam~ in~\cref{algo:edgsamp}, and state its properties.
\begin{lem}
	\label{lem:edgsamp-exp}
	Let $Y$ denote the return value of \cref{algo:edgsamp}. Then, $\expec{Y} = \bfly$.
\end{lem}
%------------------------------------------------
%------------------------------------------------
\begin{proof}
	Consider that the butterflies in $G$ are numbered from $1$ to $\bfly$.
	Let $X$ denote $\bfly_e$, when edge $e$ is chosen randomly from $E$.
	For $i = 1, \ldots, \bfly$, let $X_i$ be an indicator random variable equal to $1$ if edge $e$ is contained in butterfly $i$, and $0$ otherwise.
	We have $X = \sum_{i=1}^{\bfly} X_i$. Since each butterfly has exactly four edges in $E$, we have $\prob{X_i=1} = 4/m$.
	\begin{align*}
		\expec{X} = \sum_{i=1}^{\bfly} \expec{X_i} = \sum_{i=1}^{\bfly} \prob{X_i=1} = \frac{4\bfly}{m}
	\end{align*}
	Since $Y = X \cdot \frac{m}{4}$, it follows that $\expec{Y} = \bfly$.
\end{proof}
%------------------------------------------------
Let $p_E$ be the number of pairs of butterflies that share at least one edge. Then, $p_E = p_{1e}+p_{1w}$.
%------------------------------------------------------------------------------------------------------------------------------------------------------------------------------------------------
\begin{lem}
	$\var{Y} \le \frac{m}{4}(\bfly + p_E)$
\end{lem}
%------------------------------------------------------------------------------------------------------------------------------------------------------------------------------------------------
%------------------------------------------------------------------------------------------------------------------------------------------------------------------------------------------------
\begin{proof}
	Proceeding similar to proof of \cref{lem:versamp-var}:
	\begin{small}
		\begin{equation*}
			\begin{aligned}
				\var{Y} &= \var{\frac{m}{4}\sum_{i=1}^{\bfly}{X_i}} = \frac{m^2}{16} \var{\sum_{i=1}^{\bfly}{X_i}} \\ %\frac{n^2}{4}\sum_{i=1}^{\bfly}\sum_{j=1}^{\bfly}{\Cov\lbr i, j\rbr } \\
				&= \frac{m^2}{16}\left[\sum_{i=1}^{\bfly}{\var{X_i}} + \sum_{i\neq{}j}{\Cov\lbr X_i, X_j\rbr }\right] \\
				%     &= \frac{m^2}{16}\left[\sum_{i = 1}^{\bfly}{\lbr \expec{X_i^2} - \expec{X_i}^2\rbr } + \sum_{i\neq{}j}{\Cov\lbr i, j\rbr }\right] \\
				&= \frac{m^2}{16}\left[\bfly  \lbr \frac{4}{m} - \frac{16}{m^2}\rbr  + \sum_{i\neq{}j}{\lbr \expec{X_iX_j} - \expec{X_i}\expec{X_j}\rbr }\right]
			\end{aligned}
		\end{equation*}
	\end{small}
	%  To compute $\sum_{i\neq{}j}{(\expec{X_iX_j} - \expec{X_i}\expec{X_j})}$ we need to consider the different types of pairs of butterflies separately.     %\item
	For a pair of butterflies $(i,j)$ of
	\begin{itemize}
		\item
		Type $0v$, $1v$, or $2v$: there is zero probability of $i$ and $j$ being counted within $\bfly_e$, and hence $\expec{X_i X_j} = 0$. $\Cov \lbr X_i, X_j \rbr = -{16}/{m^2}$
		\item
		Type $1e$: $\E[X_iX_j] = \Pr[X_i=1]\Pr[X_j=1|X_i=1] = ({4}/{m})({1}/{4}) = {1}/{m}$. $\Cov\lbr X_i,X_j \rbr= {1}/{m} - {16}/{m^2}$
		\item
		Type $1w$, $\E[X_iX_j] = \Pr[X_i=1]\Pr[X_j=1|X_i=1] = ({4}/{m})({2}/{4}) = {2}/{m}$. $\Cov\lbr X_i,X_j \rbr= {2}/{m} - {16}/{m^2}$
	\end{itemize}
	\begin{equation*}
		\begin{aligned}
			\Var[Y]  &=  \frac{m^2}{16}\bigg[\bfly \lbr \frac{4}{m} - \frac{16}{m^2}\rbr  - p_{0v}\frac{16}{m^2} - p_{1v}\frac{16}{m^2}  \\
			&  - p_{2v}\frac{16}{m^2} + p_{1e}\lbr \frac{1}{m} - \frac{16}{m^2}\rbr  + p_{1w}\lbr \frac{2}{m} - \frac{16}{m^2}\rbr \bigg]\\
			& \leq \frac{m^2}{16}\left[\frac{4}{m}\lbr \bfly + p_{1e} + p_{1w}\rbr  - {\bfly\choose{2}}\frac{16}{m^2}\right]\\
			& = \frac{m}{4}\lbr \bfly + p_{1e} + p_{1w}\rbr  - {\bfly\choose{2}} \leq \frac{m}{4}\lbr \bfly + p_E\rbr
		\end{aligned}
	\end{equation*}
\end{proof}
%------------------------------------------------------------------------------------------------------------------------------------------------------------------------------------------------
Let $Z$ be the  average of $\alpha = ({8m}\left(1 + {p_E}/{\bfly}\right))/{(\epsilon^2\bfly)}$ independent instances of $Y$.
%We have $\Var[Z] = \Var[Y] / \alpha$.
%= \Var[\frac{1}{\alpha}\sum_{i=1}^{\alpha}Y] = \frac{1}{\alpha^2}\sum_{i=1}^{\alpha}{\Var[Y]} = \frac{\Var[Y]}{\alpha}$.
Using Chebyshev's inequality:
\begin{equation*}
	\begin{aligned}
		\prob{|Z - \bfly| \geq \epsilon\bfly}  &\leq \frac{\Var[Z]}{\epsilon^2{}\bfly^2} = \frac{\Var[Y]}{\alpha\epsilon^2{}\bfly^2}
		\leq \frac{m(\bfly + p_E)}{4\alpha\epsilon^2{}\bfly^2} = \frac{1}{32}
	\end{aligned}
\end{equation*}
%We can turn the above estimator that yields an $\epsilon$ relative error with probability $1/32$ into an $(\epsilon,\delta)$ estimator by taking the median of $O(\log(1/\delta))$ estimators, using standard analysis (see for e.g.~{}\cite{MUbook}, Chapter 4).
%---------------
\begin{lem}
	There is an algorithm that uses $\tilde{\bigO}\left(\frac{m}{\bfly}\left(1 + \frac{p_E}{\bfly} \right)\right)$ iterations of \cref{algo:wdgsamp} and yields an $(\epsilon,\delta)$-estimator of $\bfly(G)$ using time $\tilde{\bigO}\left(\frac{m^2 \Delta}{n\bfly}\left(1 + \frac{p_E}{\bfly} \right)\right)$. The additional space complexity is $\bigO(m\Delta/n)$.
\end{lem}
%------------
\begin{proof}
	Each iteration of the \edgSam~{}algorithm counts $\bfly_e$ for a randomly chosen $e$. Similar to the analysis of \verSam, the expected time of an iteration is $\bigO(m\Delta/n)$. Hence, the runtime follows. The space complexity is the space taken by \cref{algo:local-edgBFC}, which is $\bigO(m\Delta/n)$.
\end{proof}

\subsection{Wedge Sampling (Algorithm~\wdgSam)}

%----------------------------------------------------------------------------------------------------------------------------------------------------------------------------------
\begin{algorithm}[t]
    \DontPrintSemicolon
    \SetKwInOut{Input}{Input}
    \Input{A bipartite graph $G = (V, E)$}
    \KwOut{An estimate of $\bfly(G)$}
    $\wdg \gets \sum_{u \in V}{{d_u\choose {2}}}$~\tcp*{number of wedges}
    Choose a vertex $u \in V$ with probability ${{d_u\choose {2}}} / \wdg$\;
    Choose two distinct vertices $v, w \in \Gamma_u$ uniformly at random\;
    $\beta \gets |\Gamma_v \cap \Gamma_w| - 1$~\tcp*{\# butterflies that has $(u, v, w)$}
    \Return $\beta \cdot \wdg/4$\;
    \caption{\wdgSam~(single iteration)}
    \label{algo:wdgsamp}
\end{algorithm}
%----------------------------------------------------------------------------------------------------------------------------------------------------------------------------------

In \wdgSam, we first choose a random ``wedge", a path of length two in the graph.
This already yields three vertices that can belong to a potential butterfly.
Then we count the number of butterflies that contain this wedge by finding the intersection of the neighborhoods of the two endpoints of the wedge. %vertices that are not connected in the wedge.
\cref{algo:wdgsamp} describes the \wdgSam~{}algorithm.
Let $Y_W$ denote the return value of \cref{algo:wdgsamp}. 

\begin{lem}
	The wedge $(x, y, z)$ in lines 2 and 3 of \cref{algo:wdgsamp} is chosen uniformly at random from the set of all wedges in $G$.
\end{lem}
\begin{proof}
	In order to choose the wedge $(x,y,z)$, vertex $y$ must be chosen in step (2), followed by vertices $x$ and $z$ in step (3).
	The probability is $\frac{{d_y \choose 2}}{h} \cdot \frac{1}{{d_y \choose 2}} = 1/h$, showing that a wedge is
	sampled uniformly at random.
	%\red{$\prt{{{d_x\choose {2}}} / \wdg}\prt{1 / {{d_x\choose {2}}}} = 1/\wdg$}. The idea comes from \cite{SePiKo14}.
\end{proof}
%----------------------------------------------------------------------------------------------------------------------------------------------------------------------------------
%-----------------------------------------------------------------------------------------------------------------------------------------
\begin{lem}
	\label{lem:wdgsamp-exp}
	Let $Y$ denote the return value of \cref{algo:wdgsamp}. Then, $\expec{Y} = \bfly$
\end{lem}
%------------------------------------------------
%------------------------------------------------
\begin{proof}
	Consider that the butterflies in $G = (V, E)$ are numbered from $1$ to $\bfly$.
	Suppose that $\wdg$ is the number of wedges in $G$.
	For $i = 1, \ldots, \bfly$, let $X_i$ be an indicator random variable equal to $1$ if the $i^{\text{th}}$ butterfly includes wedge $w$.
	Let $X$ denote $\beta$, the number of butterflies that contain the sampled wedge $(x,y,z)$.
	We have $X = \sum_{i=1}^{\bfly} X_i$. Since each butterfly has four wedges in $G$, we have $\prob{X_i=1} = 4/\wdg$. The rest of the proof
	is similar to that of \cref{lem:versamp-exp}, and is omitted.
	%  \begin{equation*}
	%  \begin{aligned}
	%        \expec{X} = \sum_{i=1}^{\bfly} \expec{X_i} = \sum_{i=1}^{\bfly} \prob{X_i=1} = \frac{4\bfly}{\wdg}
	%  \end{aligned}
	%  \end{equation*}
	%Since $Y = X \times \frac{\wdg}{4}$, it follows that $\expec{Y} = \bfly$.
\end{proof}

%Let $p_w$ denote the number of pairs of butterflies that share a wedge.
%------------------------------------------------------------------------------------------------------------------------------------------------------------------------------------------------
\begin{lem}
	$\var{Y} \le \frac{\wdg}{4}(\bfly + p_{1w})$
\end{lem}
%------------------------------------------------------------------------------------------------------------------------------------------------------------------------------------------------
%------------------------------------------------------------------------------------------------------------------------------------------------------------------------------------------------
\begin{proof}
	\begin{small}
		\begin{equation*}
			\begin{aligned}
				\var{Y} &= \var{\frac{\wdg}{4}\sum_{i=1}^{\bfly}{X_i}} = \frac{\wdg^2}{16} \var{\sum_{i=1}^{\bfly}{X_i}} \\ %\frac{n^2}{4}\sum_{i=1}^{\bfly}\sum_{j=1}^{\bfly}{\Cov\lbr i, j\rbr } \\
				%     &= \frac{\wdg^2}{16}\left[\sum_{i=1}^{\bfly}{\var{X_i}} + \sum_{i\neq{}j}{\Cov\lbr X_i, X_j\rbr }\right] \\
				%     &= \frac{\wdg^2}{16}\left[\sum_{i = 1}^{\bfly}{\lbr \expec{X_i^2} - \expec{X_i}^2\rbr } + \sum_{i\neq{}j}{\Cov\lbr i, j\rbr }\right] \\
				&= \frac{\wdg^2}{16}\left[\bfly \lbr \frac{4}{\wdg} - \frac{16}{\wdg^2}\rbr  + \sum_{i\neq{}j}{\lbr \expec{X_iX_j} - \expec{X_i}\expec{X_j}\rbr }\right]
			\end{aligned}
		\end{equation*}
	\end{small}
	% To compute $\sum_{i\neq{}j}{(\expec{X_iX_j} - \expec{X_i}\expec{X_j})}$: {\iffalse we need to consider the different types of pairs of butterflies separately. \fi}    %\item
	For a pair of butterflies $(i,j)$ of
	\begin{itemize}
		\item
		Type $0v$, $1v$, $2v$, or $1e$, there is zero probability of $i$ and $j$ being counted together, and hence $\expec{X_i X_j}  = 0$. We have $\Cov \lbr X_i, X_j \rbr = -{16}/{\wdg^2}$
		\item
		Type $1w$, $\E[X_iX_j] = \Pr[X_i=1]\Pr[X_j=1|X_i=1] = ({4}/{\wdg})({1}/{4}) = {1}/{\wdg}$. Therefore, $\Cov\lbr X_i,X_j \rbr= {1}/{\wdg} - {16}/{\wdg^2}$.
	\end{itemize}
	\begin{equation*}
		\begin{aligned}
			\Var[Y]  &=  \frac{\wdg^2}{16}\bigg[\bfly \lbr \frac{4}{\wdg} - \frac{16}{\wdg^2}\rbr  - p_{0v}\frac{16}{\wdg^2} - p_{1v}\frac{16}{\wdg^2}  \\
			&  - p_{2v}\frac{16}{\wdg^2} - p_{1e}\frac{16}{\wdg^2} + p_{1w}\lbr \frac{1}{\wdg} - \frac{16}{\wdg^2}\rbr \bigg]\\
			& \leq \frac{\wdg^2}{16}\left[\frac{4}{\wdg}\lbr \bfly + p_{1w}\rbr  - {\bfly\choose{2}}\frac{16}{\wdg^2}\right]\\
			& = \frac{\wdg}{4}\lbr \bfly + p_{1w}\rbr  - {\bfly\choose{2}} \leq \frac{\wdg}{4}\lbr \bfly + p_{1w}\rbr
		\end{aligned}
	\end{equation*}
\end{proof}
%------------------------------------------------------------------------------------------------------------------------------------------------------------------------------------------------
Let $Z$ be the  average of $\alpha = ({8\wdg}\left(1 + {p_{1w}}/{\bfly}\right))/{(\epsilon^2\bfly)}$ independent instances of $Y$.
Using Chebyshev's inequality, we arrive that $Z$ is an $(\epsilon, 1/32)$-estimator of $\bfly(G)$.
%We have $\Var[Z] = \Var[Y] / \alpha$.
%= \Var[\frac{1}{\alpha}\sum_{i=1}^{\alpha}Y] = \frac{1}{\alpha^2}\sum_{i=1}^{\alpha}{\Var[Y]} = \frac{\Var[Y]}{\alpha}$.
%Using Chebyshev's inequality, we have:
%  \begin{equation*}
%  \begin{aligned}
%     \prob{|Z - \bfly| \geq \epsilon\bfly}  &\leq \frac{\Var[Z]}{\epsilon^2{}\bfly^2} = \frac{\Var[Y]}{\alpha\epsilon^2{}\bfly^2}
%    \leq \frac{\wdg(\bfly + p_{1w})}{4\alpha\epsilon^2{}\bfly^2} = \frac{1}{32}
%  \end{aligned}
%  \end{equation*}

%We can turn the above estimator that yields an $\epsilon$ relative error with probability $1/32$ into an $(\epsilon,\delta)$ estimator by taking the median of $O(\log(1/\delta))$ estimators, using standard analysis (see for e.g.~{}\cite{MUbook}, Chapter 4).
%---------------
\begin{lem}
	There is an algorithm that uses $\allowbreak\tilde{\bigO}\bigl(\frac{\wdg}{\bfly}\bigl(1 + \frac{p_{1w}}{\bfly} \bigr)\bigr)$
	iterations of \cref{algo:wdgsamp} and yields an $(\epsilon,\delta)$-estimator of $\bfly(G)$ using time $\tilde{\bigO}\left(\frac{(\Delta + \log n)\wdg}{\bfly}\left(1 + \frac{p_{1w}}{\bfly} \right)\right)$ and space $\bigO(n)$.
\end{lem}
%------------
\begin{proof}
	We first consider the runtime of \cref{algo:wdgsamp}. At first glance, it looks like Steps (1) and (2) take $\bigO(n)$ time.
	This can be reduced to $\bigO(\log n)$ using a pre-computation step of time $\bigO(n)$, which precomputes the value of $h$, and also stores
	the values of ${d_v \choose 2}$ for different vertices $v$ in an array $A[1\ldots n]$ of length $n$. We consider another array $B[1\ldots n]$, where each element $B[j] = \sum_{i=1}^j A[i]$.  When \wdgSam~{}is called, a random number $r$ in the range $\{1,2,\ldots,h\}$ is first generated, and Step (2) is accomplished through finding the smallest $j$ such that $B[j] \ge r$. Since array $B$ is sorted in increasing order, this can be done in time $\bigO(\log n)$ using a binary search. Steps (3) and (4) of \cref{algo:wdgsamp} can be performed in $\bigO(\Delta)$ time -- Step (4), for example, through storing all elements of $\Gamma_x$ in a hash set, and repeatedly querying for elements of $\Gamma_y$ in this hash set, for a total of $\bigO(d_y) = \bigO(\Delta)$ time. The additional space required by this algorithm, in addition to the stored graph itself, is $\bigO(n)$, for random sampling.
\end{proof}

%----------------------------------------------------------------------------
%\subsection{Comparison of Sampling Algorithms}
\subsection{Accuracy and Runtime of Sampling}
%----------------------------------------------------------------------------

\paragraph*{Accuracy of a Single Iteration.}
In order to understand the relation between the three sampling algorithms, we compare the variance of the estimates retuned by these algorithms. Note that each of them returns an unbiased estimate of the number of butterflies. The standard deviations (square root of variances) of \verSam, \edgSam, and \wdgSam~{}for different graphs are estimated and summarized in \cref{table:variance-sam}. The results show that the variances of \verSam~{}and \wdgSam~{}are much higher than the variance of \edgSam. Note that these are estimates of an upper bound on the variances, and the actual variances could be (much) smaller.

%	\erdem{Errors are sqrt of variance numbers over the num of butterflies. For vertex and wedge those upper bound are correct. but i didn't understand it for edge sampling? We do have errors more than those numbers (in figs 5-8). any idea?}}
%	\vspace*{-3ex}

\begin{table}[t!]
\footnotesize
%\vspace{-1ex}
\renewcommand{\tabcolsep}{2.5pt}
\centering
\begin{tabular}{|c|c|r|c|r|c|r|c|}\hline
& \multicolumn{2}{c|}{\verSam} & \multicolumn{2}{c|}{\edgSam} & \multicolumn{2}{c|}{\wdgSam} \\
& $\sqrt{\frac{(\bfly+p_V)n}{4}}$ & error \%      & $\sqrt{\frac{(\bfly+p_E)m}{4}}$ & error \%    & $\sqrt{\frac{(\bfly+p_{1w})h}{4}}$ & error \% \\ \hline
\deli	&	$	2.8\times 10^{13}	$	&	$	494.9	$	&	$	2.1 \times 10^{12}	$	&	$	38.3	$	&	$	7.7\times 10^{11}	$	&	$	13.6	$	\\ \hline
\jrn	&	$	2.3\times 10^{15}	$	&	$	708.56	$	&	$	2.1 \times 10^{13}	$	&	$	6.4	$	&	$	1.4\times 10^{14}	$	&	$	99.3	$	\\ \hline
\ork	&	$	5\times 10^{15}	$	&	$	223.6	$	&	$	2 \times 10^{14}	$	&	$	9.2	$	&	$	2.9\times 10^{14}	$	&	$	13.3	$	\\ \hline
\web	&	$	6.8\times 10^{16}	$	&	$	3431.5	$	&	$  8.2 \times 10^{13}	$	&	$	4.1	$	&	$	4.3\times 10^{16}	$	&	$	2194.8	$	\\ \hline
\wiki	&	$	1.3\times 10^{16}	$	&	$	6823.6	$	&	$	3.8 \times 10^{13}	$	&	$	19	$	&	$	1.4\times 10^{15}	$	&	$	703.4	$	\\ \hline
	\end{tabular}
	\caption{Standard deviations of estimators on large graphs. For each method, the theoretical upper bound on the standard deviation is shown in the left column, and the ``error", the ratio between the (theoretical) standard deviation and the number of butterflies, is shown in the second column.
%	\erdem{Vahid -- Did you update all the edge sampling numbers here? They're all different than the arxiv version. -- yes, these are current numbers}
	}
\label{table:variance-sam}
\vspace{-4ex}
\end{table}

%Let us compare the variances of \verSam~{}and \edgSam.  
The variance of \verSam~{}is proportional to $n p_V$ where $p_V$ is the number of {\em pairs of butterflies that share a vertex} and $n$ is the number of vertices. The variance of \edgSam~{} is proportional to $m p_E$ where $p_E$ is the number of {\em pairs of butterflies that share an edge} and $m$ is the number of edges. Note that typically $\bfly \ll p_V, p_E$ and hence $\bfly + p_V \approx p_V$ and $\bfly + p_E \approx p_E$.
If two butterflies share an edge, they certainly share a vertex, hence $p_E \le p_V$. Since it is possible that two butterflies share a vertex but do not share an edge, $p_E$ could be much smaller than $p_V$. It turns out that in most of these graphs, $p_E$ was much smaller than $p_V$.
On the other hand, the number of vertices in a graph ($n$) is comparable to the number of edges ($m$). Typically $m < 10n$, and only in one case (\ork), we have $m \approx 30n$.
Thus, $n p_V \ll m p_E$ for the graphs we consider. 
As a result, the variance of \verSam~{} is much larger than the variance of \edgSam{}, which is reflected clearly in \cref{table:variance-sam}. 
%This explains why \verSam~{}performs poorly when compared with \edgSam. 

Comparing \edgSam~{}with \wdgSam, we note that the variance of \wdgSam~{}is proportional to $\wdg \cdot p_{1w}$, where $\wdg$ is the number of wedges in the graph ($\bigO(\sum_v d_v^2)$) and $p_{1w}$ is the number of pairs of butterflies that share a wedge.
The variance of \edgSam~ is proportional to $m p_E$. $p_{1w} \le p_E$ since each pair of butterflies that shares a wedge also shares an edge.
At the same time, we see that $\wdg$ is substantially greater than $m$.
Overall, there is no clear winner among \wdgSam~ and \edgSam~in theory, but \edgSam~ seems to have the smaller variance on real-world networks, typically, sometimes much smaller, as in graph \wiki.
%the  variance of \wdgSam~{}is larger than that of \edgSam, and \edgSam~{}has the smallest variance among all.
%Due to high variance on large graphs, \verSam and \wdgSam~{}do not converge to less than one percent error. In the graph \deli, the variance of \wdgSam~{}is not too much. \cref{figure:deli-sam-err-time} shows that this algorithm estimates the number of butterflies within one percent error after a few seconds.

%------------------------------------------------

%------------------------------------------------

%------------------------------------------------
\begin{figure}[!b]
\centering
\vspace{-1ex}
%newResults/histogram/time-per-sample
\includegraphics[width=\linewidth]{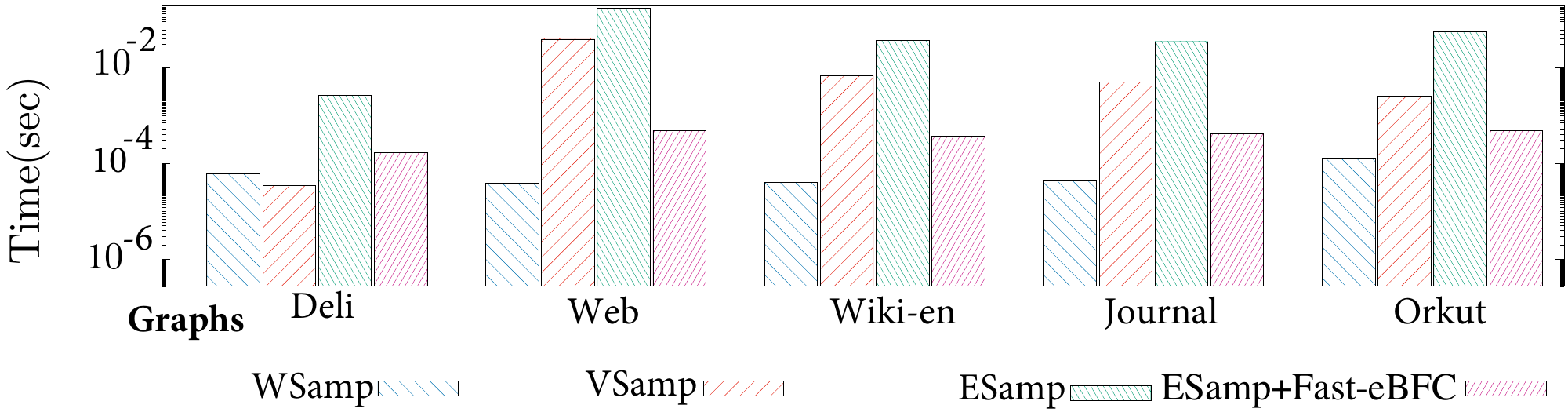}
\vspace{-5ex}
\caption{{Average time per iteration of sampling. For \edgSam~with \fstEdg, the time is shown for 1000 iterations of \fstEdg~(\cref{algo:qBFC}).}} %\erdem{Vahid -- Can you extend the y-axis until $10^0$? -- The plots were labeled incorrectly, and they are correct now.}}}
% Seems like the FastSamp results are obtained by 1000 iterations of Alg 7; make sure to note that here in caption}}}%\erdem{make it smaller and also remove the fast edge sampling bars (will be shown in~\cref{sec:quickBFC})}}
\label{figure:timePerSample}
\vspace{-1ex}
\end{figure}

\begin{figure*}[t!]
	\centering
%	\vspace{-3ex}
	\includegraphics[width=\textwidth]{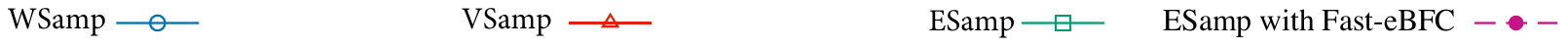}
\end{figure*}

\begin{figure*}[t!]
	\centering
	\vspace{-3ex}
	\resizebox{1\linewidth}{!}{
	\subfloat[][{\deli}]{\includegraphics[width=0.19\textwidth]{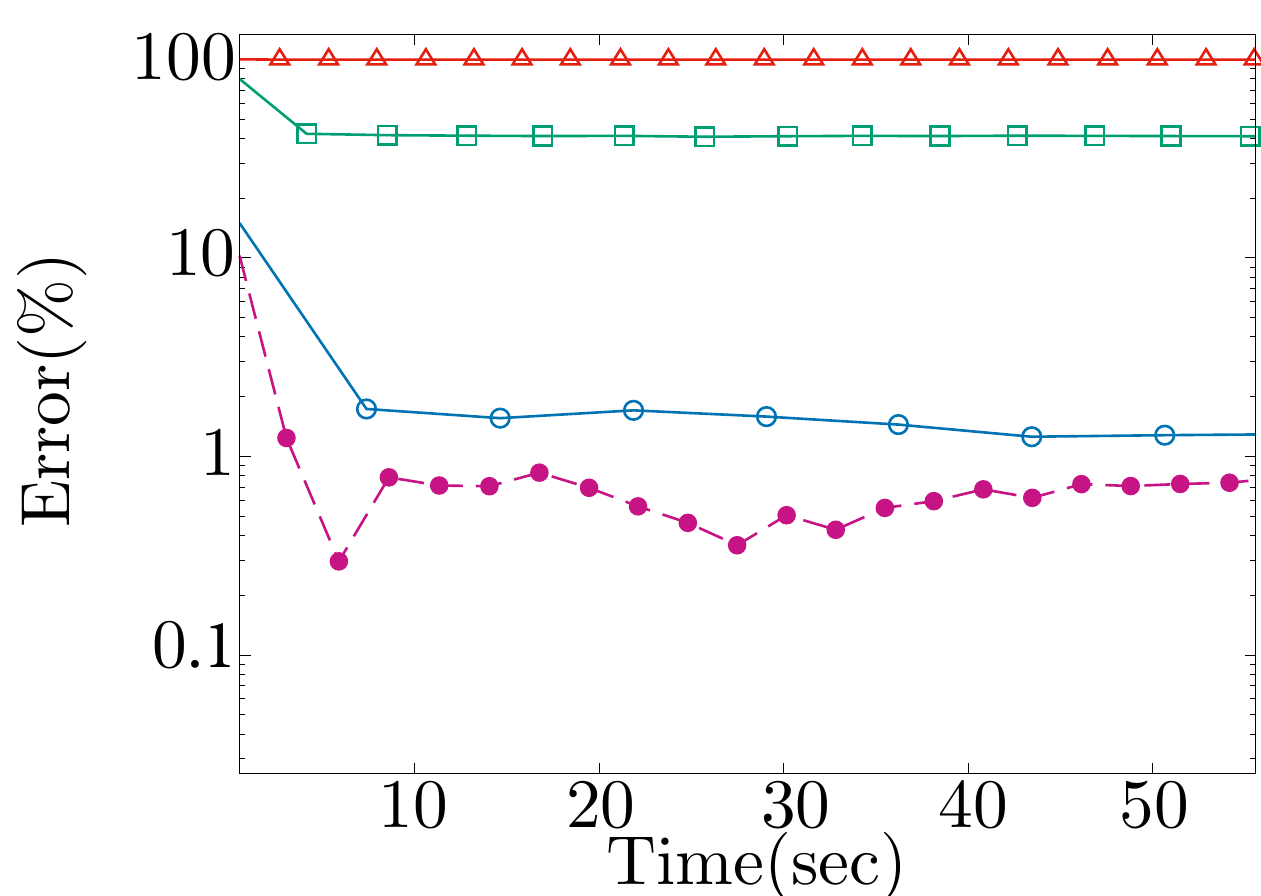}\label{figure:deli-sam-err-time}}~{}~{}
	\hspace{1ex}
	\subfloat[][{\jrn}]{\includegraphics[width=0.19\textwidth]{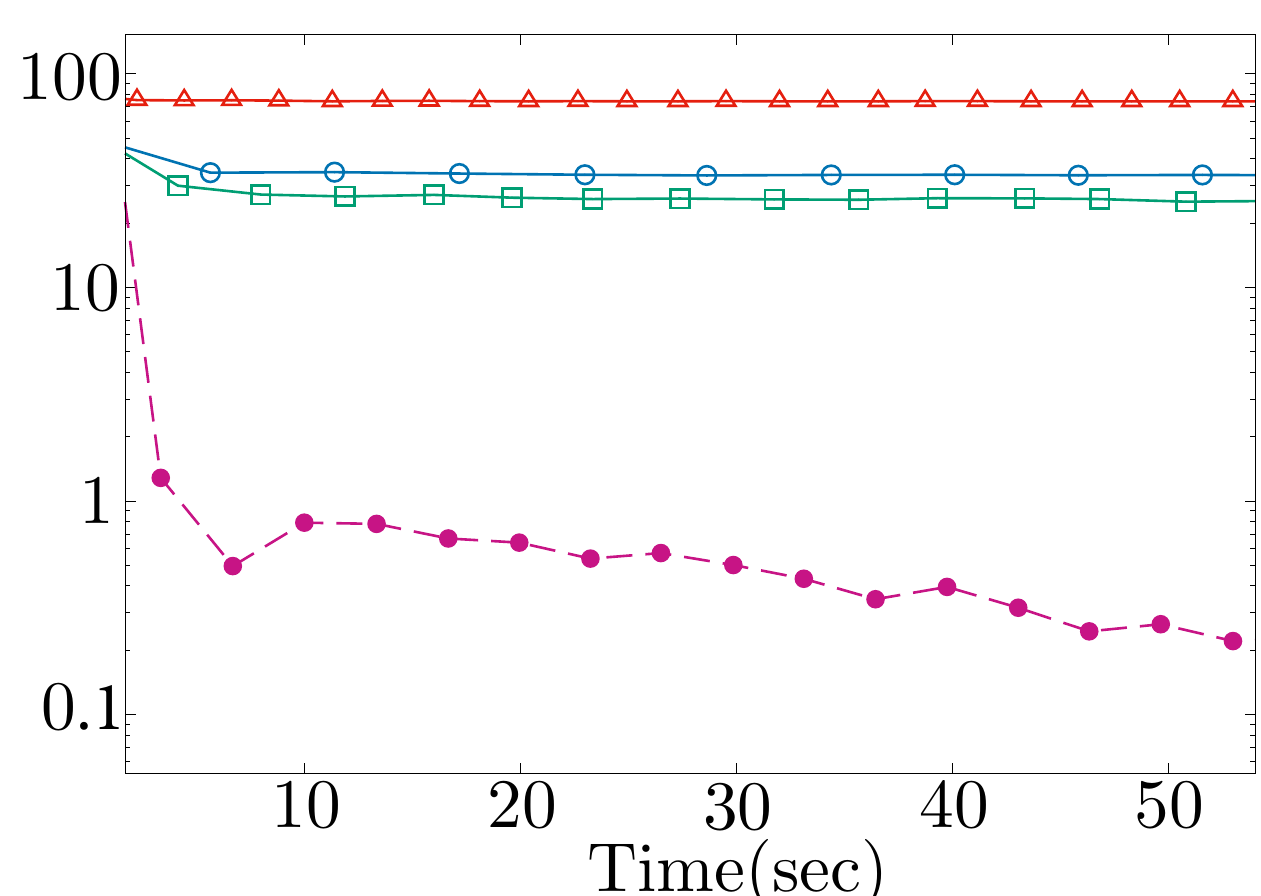}\label{figure:jrn-sam-err-time}}~{}~{}
	\hspace{1ex}
	\subfloat[][\ork]{\includegraphics[width=0.19\textwidth]{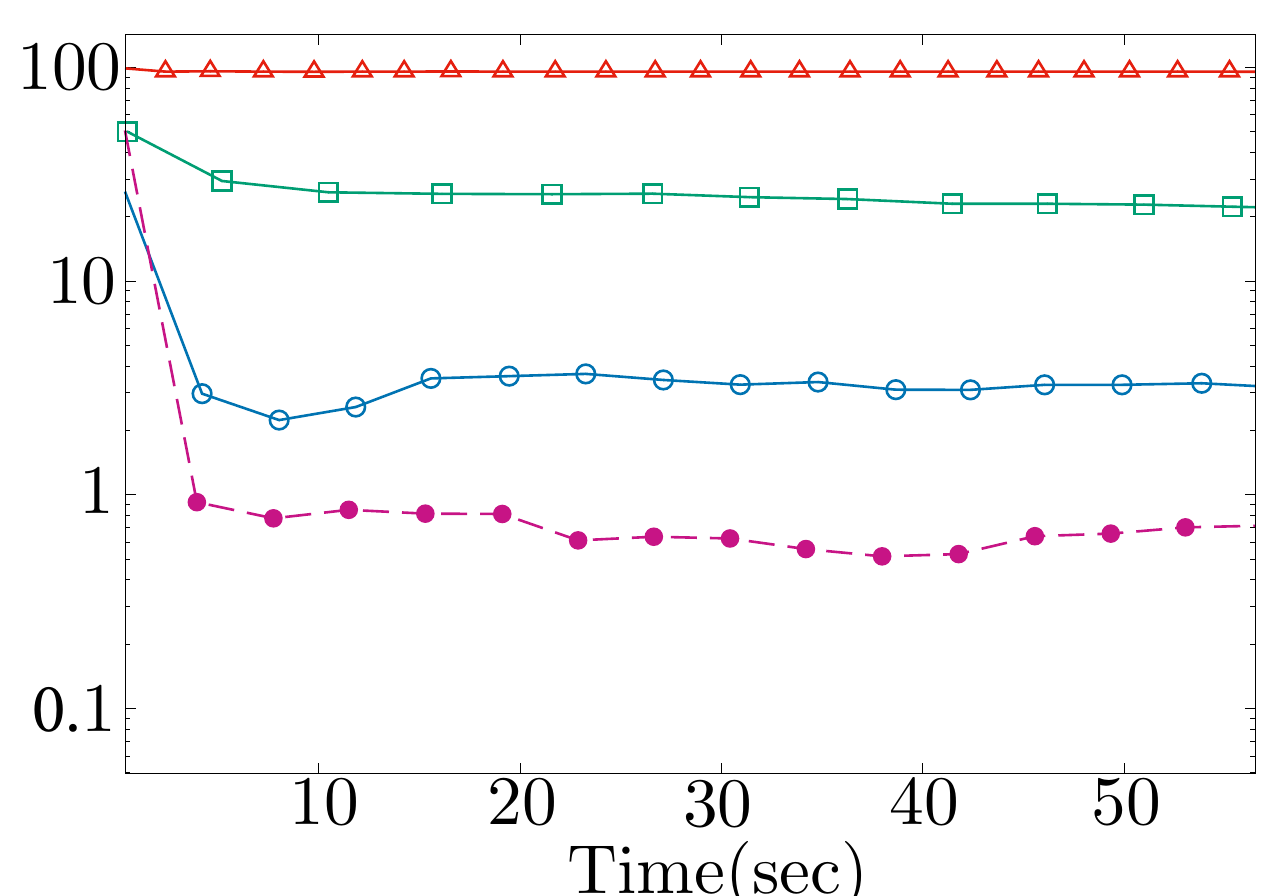}\label{figure:ork-sam-err-time}}~{}~{}
	\hspace{1ex}
	\subfloat[\web]{\includegraphics[width=0.19\textwidth]{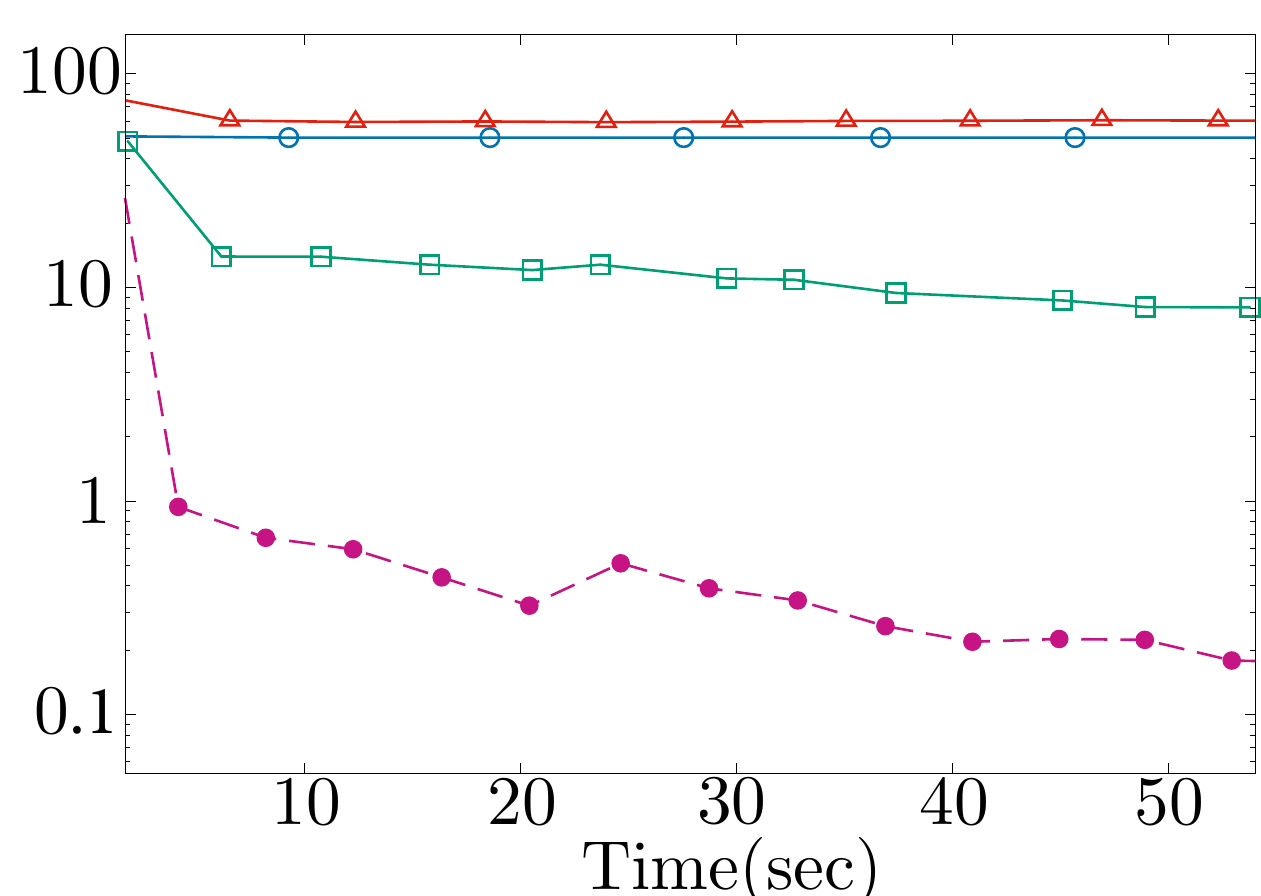}\label{figure:web-sam-err-time}}~{}~{}
	\hspace{1ex}
	\subfloat[\wiki]{\includegraphics[width=0.19\textwidth]{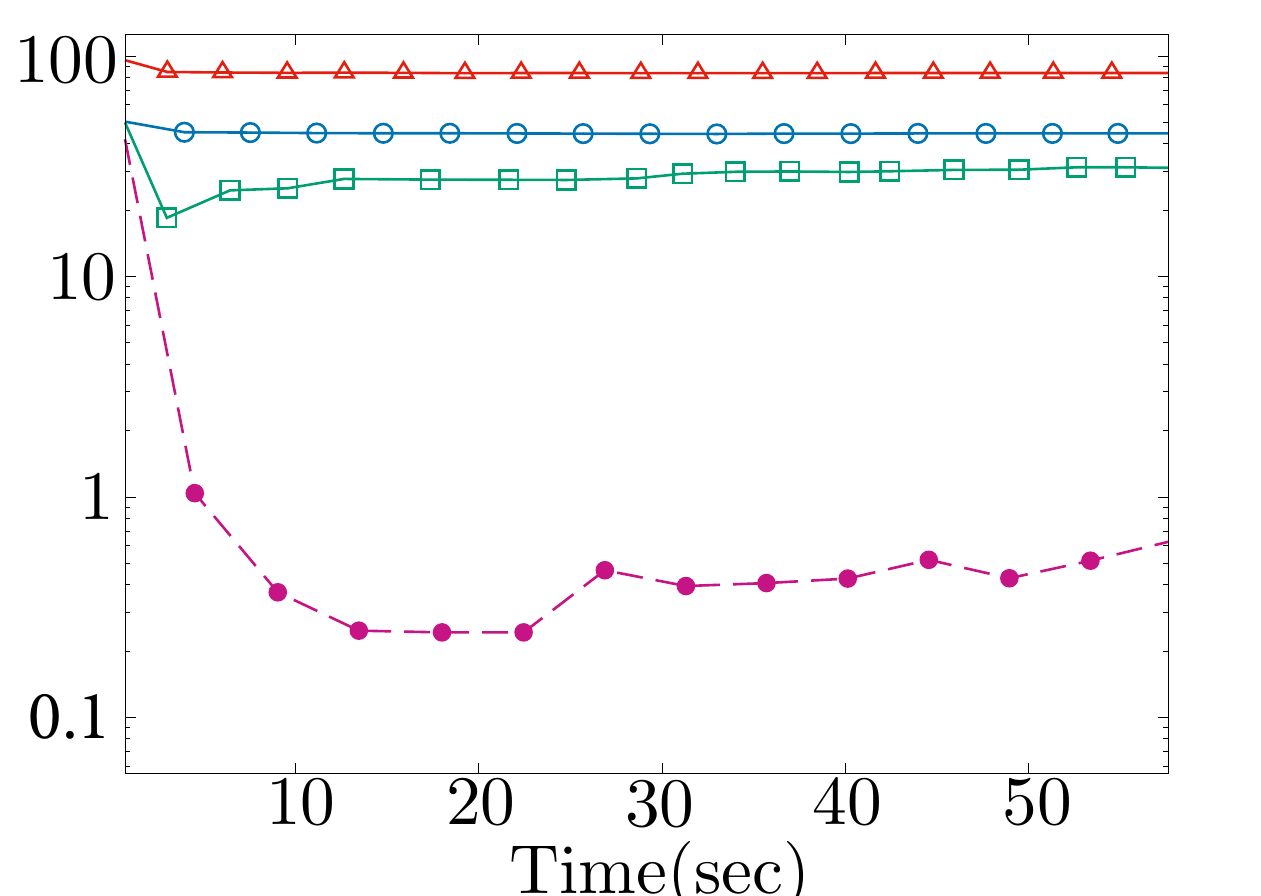}\label{figure:wiki-sam-err-time}}~{}~{}
	}
	\vspace{-1ex}
	\caption{ \bf {Relative error as a function of runtime, for sampling algorithms. \edgSam~with \fstEdg~yields $< 1\%$ relative error within $5$ seconds for all networks.}%\erdem{can we remove Error \% label in the right-most -- done}
	}\label{figure:samp}
	\vspace{-3ex}
\end{figure*}

\vspace*{-2ex}
\paragraph*{Runtime per Iteration:}
The performance of a sampling algorithm depends not only on the variance of an estimator, but also on how quickly an estimator can be computed. \cref{figure:timePerSample} shows the time taken to compute a single estimator using different sampling algorithms for the five largest networks. We note that \edgSam~(which calls \peredge) requires the largest amount of time per sampling step, among all estimators. This decreases the overall accuracy of \edgSam, despite its smaller variance.
% observed in~\cref{table:variance-sam}. 
%However, \edgSam~{}requires the largest amount of time per sampling step with a significant difference. 

%------------------------------------------------
%\subsection{Fast Butterfly Counting: \fstEdg}
\subsection{Faster Edge Sampling using \fstEdg}
\label{sec:quickBFC}
%------------------------------------------------
Since \edgSam~had a low variance, but a high runtime per iteration, we tried to achieve different tradeoffs with respect to runtimes and accuracy, which led to our next algorithm \fstEdg, a faster variant of butterfly counting per edge (\peredge). Like \edgSam, we first sample a random edge from the graph, but instead of exactly counting the number of butterflies that contain the sampled edge, which leads to (relatively) expensive iterations, \fstEdg~only estimates the number of butterflies per edge, through a further sampling step (replacing \peredge~in \cref{ln:esmp} of \edgSam~(\cref{algo:edgsamp})). For an edge $(u, v)$ this estimation is performed by randomly choosing one neighbor each of $u$ and of $v$, and checking if the four vertices form a butterfly. 
\cref{algo:qBFC} presents a single iteration of~\fstEdg. This procedure is repeated a few times for a given edge, to improve the accuracy of the estimate. In our implementation we repeated it 1000 times for each sampled edge, and it was still significantly faster than \peredge~(\cref{figure:timePerSample}). We use it instead of the \peredge~algorithm in \edgSam.
While the estimate from each iteration is less accurate than in \edgSam, more iterations are possible within the same time.
%\cref{tab:quickBFC} shows that
 \fstEdg~ (with 1000 repetitions) is faster than \peredge~ by 15x-357x, when used in \edgSam.
%We observe speedups from $15.4$ up to $357.4$ per each iteration.
Overall, in large graphs, this leads to an improvement in accuracy over \peredge~in \edgSam. %For a given edge $e$, \cref{algo:qBFC} yields an unbiased estimate of $\bfly_e$.
Let $Y_{FE}$ denote the return value of \cref{algo:qBFC}. 
\begin{lem}
\vspace{-1ex}
\label{lem:fastedge-samp}
$\expec{Y_{FE}} = \bfly_e$ and $\var{Y_{FE}} \leq \bfly_e(d_u \cdot d_v)$.
\vspace{-1.3ex}
\end{lem}

\begin{proof}
	Consider that the butterflies include the edge $e = (u, v)$ are numbered from $1$ to $\bfly_e$.
	For $i = 1, \ldots, \bfly_e$, let $X_i$ be an indicator random variable. $X_i$ is equal $1$ if the $i^{\text{th}}$ butterfly contains both sampled neighbors $v'$ and $u'$, otherwise $0$.
	Let $X = \sum_{i=1}^{\bfly_e} X_i$. Since vertex $v'$ and $u'$ are chosen with probability $1/d_v$ and $1/d_u$ respectively, $\prob{X_i=1} = 1/(d_v \cdot d_u)$. Then,
	\begin{align*}
		\expec{X} = \sum_{i=1}^{\bfly_e} \expec{X_i} = \sum_{i=1}^{\bfly_e} \prob{X_i=1} = \frac{\bfly_e}{d_v\cdot{}d_u}
	\end{align*}
	$Y_{FE} = X \cdot d_v \cdot d_u$ follows that $\expec{Y_{FE}} = \bfly_e$.
	%  \begin{equation*}
	%  \begin{aligned}
	%        \expec{X} = \sum_{i=1}^{\bfly} \expec{X_i} = \sum_{i=1}^{\bfly} \prob{X_i=1} = \frac{4\bfly}{\wdg}
	%  \end{aligned}
	%  \end{equation*}
	%Since $Y = X \times \frac{\wdg}{4}$, it follows that $\expec{Y} = \bfly$.
\end{proof}
%Let $p_w$ denote the number of pairs of butterflies that share a wedge.
%------------------------------------------------------------------------------------------------------------------------------------------------------------------------------------------------
\begin{lem}
	$\var{Y_{FE}} \leq \bfly_e(d_v \cdot d_u)$
\end{lem}
%------------------------------------------------------------------------------------------------------------------------------------------------------------------------------------------------
%------------------------------------------------------------------------------------------------------------------------------------------------------------------------------------------------
\begin{proof}
	\begin{small}
		\begin{equation*}
			\begin{aligned}
				\var{Y_{FE}} &= \var{(d_v \cdot d_u)\sum_{i=1}^{\bfly_e}{X_i}} = (d_v \cdot d_u)^2~{}\var{\sum_{i=1}^{\bfly_e}{X_i}} \\ %\frac{n^2}{4}\sum_{i=1}^{\bfly}\sum_{j=1}^{\bfly}{\Cov\lbr i, j\rbr } \\
				%     &= \frac{\wdg^2}{16}\left[\sum_{i=1}^{\bfly}{\var{X_i}} + \sum_{i\neq{}j}{\Cov\lbr X_i, X_j\rbr }\right] \\
				%     &= \frac{\wdg^2}{16}\left[\sum_{i = 1}^{\bfly}{\lbr \expec{X_i^2} - \expec{X_i}^2\rbr } + \sum_{i\neq{}j}{\Cov\lbr i, j\rbr }\right] \\
				%     &= (d_v \cdot d_u)^2\left[\sum_{i=1}^{\bfly_e}{\left[\expec{X_i^2} - \expec{X_i}^2\right]} + \sum_{i \neq j}{\Cov(X_i, X_j)}\right]\\
				&= (d_v \cdot d_u)^2\left[\sum_{i=1}^{\bfly_e}{\left[\frac{1}{d_v \cdot d_u} - \frac{1}{(d_v \cdot d_u)^2}\right]} - {\bfly_e \choose {2}}\right]\\
				&\le \bfly_e \cdot d_v \cdot d_u
			\end{aligned}
		\end{equation*}
	\end{small}
	% To compute $\sum_{i\neq{}j}{(\expec{X_iX_j} - \expec{X_i}\expec{X_j})}$: {\iffalse we need to consider the different types of pairs of butterflies separately. \fi}    %\item
\end{proof}

\begin{algorithm}[t!]
    \DontPrintSemicolon
    \KwIn{An edge $e = (v, u) \in E$ in $G = (V, E)$}
    \KwOut{An estimate of $\bfly_e$}
    Choose a vertex $w$ ($\ne v$) from $\Gamma_u$ uniformly at random\;
    Choose a vertex $x$ ($\ne u$) from $\Gamma_v$ uniformly at random\;
    \IfThen{$(u, v, w, x)$ forms a butterfly}{${\beta} \gets 1$}\;
    \lElse {${\beta} \gets 0$}
    \Return ${\beta}\cdot d_u \cdot d_v$\; %\erdem{Vahid said those should be $d_u$ and $d_v$, as in the old version?}\;
    \caption{\fstEdg~(replaces \peredge~in \edgSam)
    \label{algo:qBFC}}
\end{algorithm}

Let $Z$ be the  average of $\alpha = 32({d_u\cdot{}d_v})/{(\epsilon^2\bfly_e)}$ independent instances of $Y_{FE}$.
Using Chebyshev's inequality, we arrive that $Z$ is an $(\epsilon, 1/32)$-estimator of $\bfly_e$.

{\bf Fast Wedge Sampling:} We also experimented with a faster version of \wdgSam, where the number of butterflies containing a wedge was estimated using sampling. But this did not give good results, partly because the time for each iteration of \wdgSam~was already relatively small. For example, in \deli, after 5 secs, \wdgSam{}~have less than 2\% error while Fast Wedge Sampling ended up with 29\% error after 10 secs. In \web{}~also the error percentage of Fast Wedge Sampling is $60\%$ higher than \wdgSam's. %\sri{Give some numbers here?}

\subsection{Comparing Sampling-based Approaches}
We compare the sampling algorithms \verSam, \edgSam + \peredge, \edgSam + \fstEdg, and \wdgSam. A sampling algorithm immediately starts producing estimates that get better as more iterations are executed. We record the number of iterations and relative percent error of sampling algorithms up to 60 seconds. \cref{figure:samp} shows the relative percent error with respect to the runtime for five large bipartite graphs. We report the median error of the 30 trials of sampling methods for each data point.

Our experiments show that \verSam~{}performs poorly when compared with \edgSam~{}and \wdgSam~under the same time budget -- this is along expected lines, based on our analysis of the variance. The accuracy of \edgSam~ and \wdgSam~ are comparable, with \edgSam~ being slightly better.  Note that \edgSam~has a better variance, while \wdgSam~has faster iterations.
%As we show in~\cref{table:variance-sam}, \edgSam~gives way better variance results than \wdgSam. However, since \wdgSam~is faster than~\edgSam per iteration, more iterations of \wdgSam~are possible given the same time budget.
For instance, on the \wiki~network, 1000 iters of \wdgSam~ yields $51\%$ error in $0.13$ secs, whereas 1000 iters of \edgSam~yields $29\%$ error in $33$ seconds.  %The former takes $0.13$ seconds while the latter needs $33.76$ secs.
Across all the graphs, \edgSam~using \fstEdg~, which combines the benefits of faster iterations with a good variance, yields the best results, and is superior to all other sampling methods; \textbf{it leads to less than 1\% relative error within $5$ seconds}.
%\erdem{any explanation for that? Deli, orkut, and wiki are like that. Oh ok. This is because we show the time in x-axis. We KNOW that if the same number of samples are used, edge sampling definitely outperforms wedge sampling -- I also checked the raw data. But we need to state this explicitly.}

%------------------------------------------------
%\section{Approximate Counting Using Sparsification}
\section{Approximation by One-Shot Sparsification}
\label{sec:sparsify}
\newtheorem{obs}{Observation}
\newcommand{\ncolors}{N}
%------------------------------------------------
We present methods for estimating $\bfly(G)$ using \textit{one-shot sparsification} -- where we thin down the input graph into a smaller graph through a global sampling step. The number of butterflies in the sparsified graph is used to estimate $\bfly(G)$. Unlike algorithms such as \edgSam~and \wdgSam, which work on a small subgraph constructed around a single randomly sampled edge or a wedge, sparsification methods put each edge in the graph into the sample with a certain probability. We consider two approaches to sparsification (1) \edgSpr ~in~\cref{sec:edg-sparsify}, where edges are chosen {\em independently} of each other and (2) \clr~in \cref{sec:col-sparsify}, based on a sampling method where different edges are not independently chosen, but dense regions appear with higher probability in the sample -- based on a similar idea in the context of triangle counting due to Pagh and Tsourakakis~\cite{colorful-triangles}.

\begin{figure*}[t!]
	\centering
%	\vspace{-3ex}
	\includegraphics[width=\textwidth]{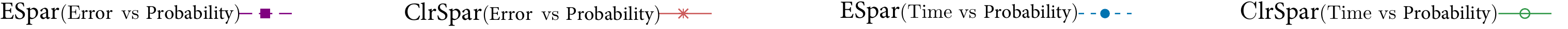}
%	\vspace{-6ex}
\end{figure*}
\begin{figure*}[t!]
\centering
	\vspace{-3ex}
	\resizebox{1\linewidth}{!}{
	\subfloat[][{\deli}]{\includegraphics[width=0.33\textwidth]{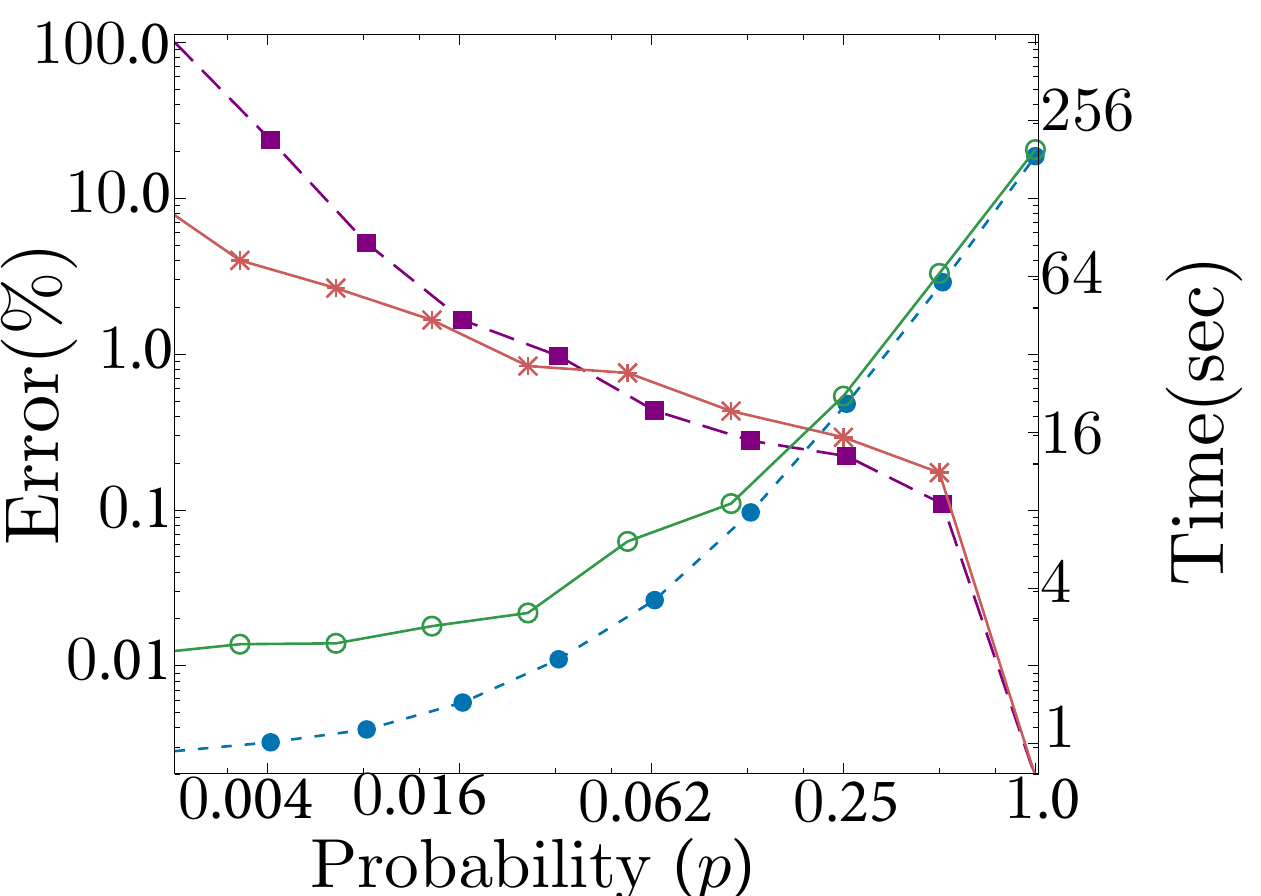}\label{figure:deli-sprs-err-prob}}~{}~{}
	\hspace{1ex}
	\subfloat[][{\jrn}]{\includegraphics[width=0.33\textwidth]{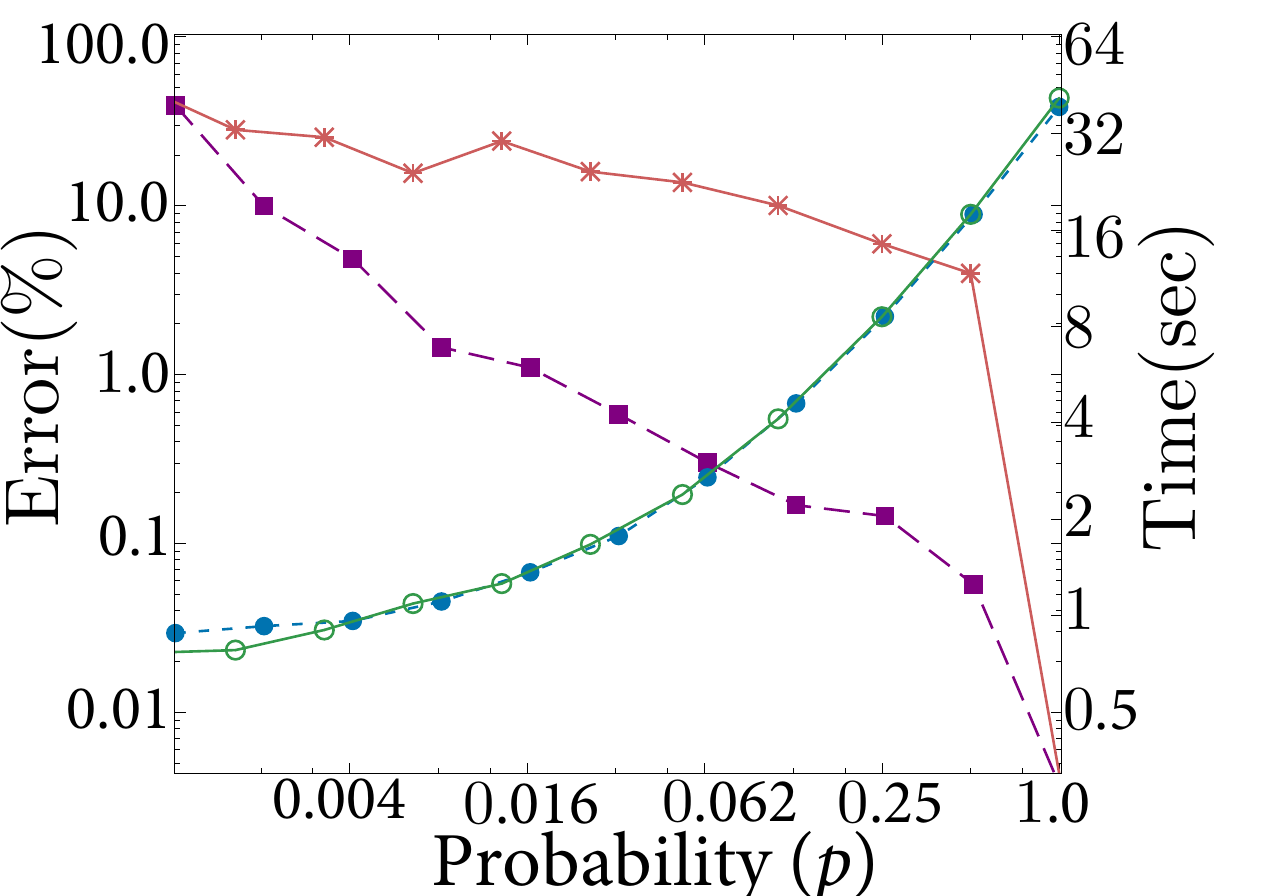}\label{figure:jrn-sprs-err-prob}}~{}~{}	
	\hspace{1ex}
	\subfloat[][{\ork}]{\includegraphics[width=0.33\textwidth]{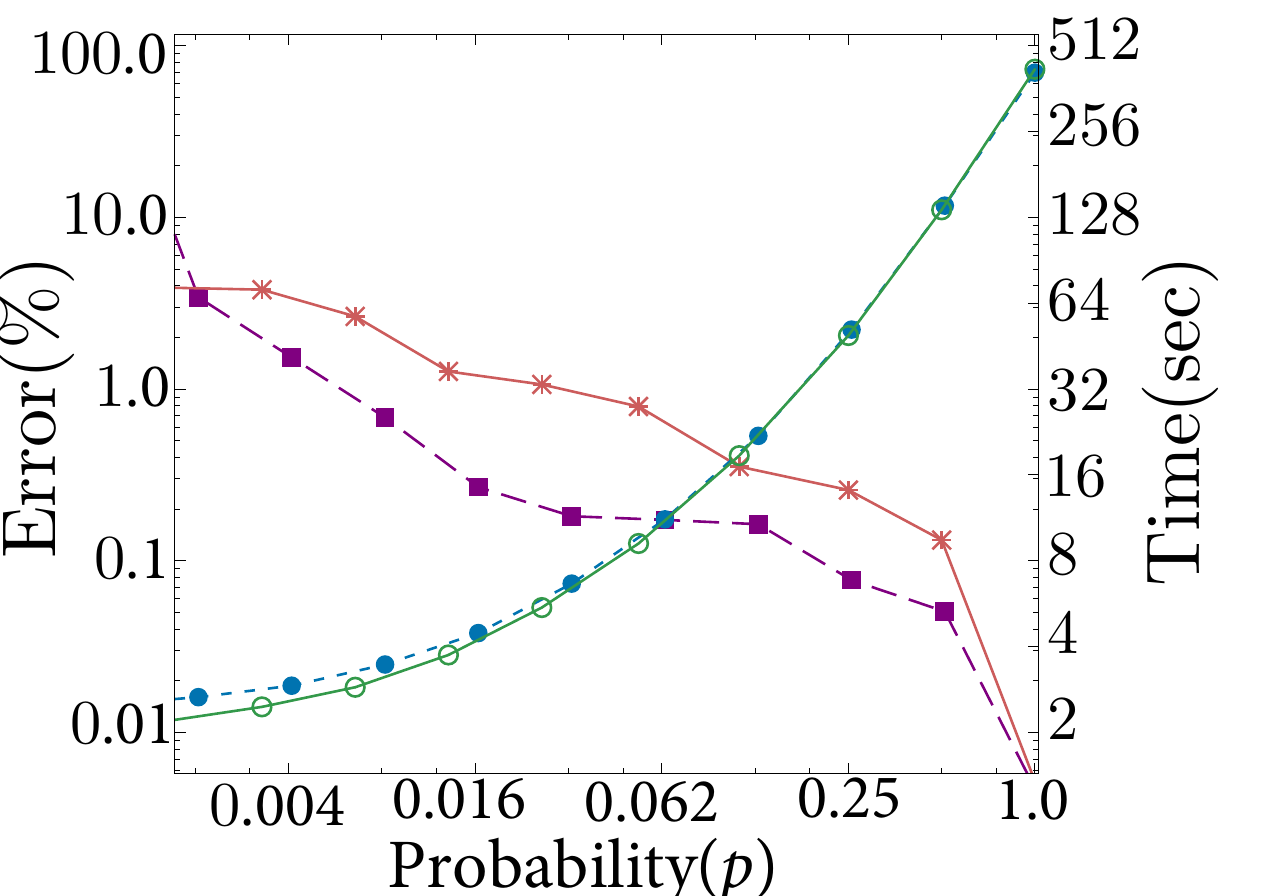}\label{figure:ork-sprs-err-prob}}~{}~{}
	\hspace{1ex}
	\subfloat[][{\web}]{\includegraphics[width=0.33\textwidth]{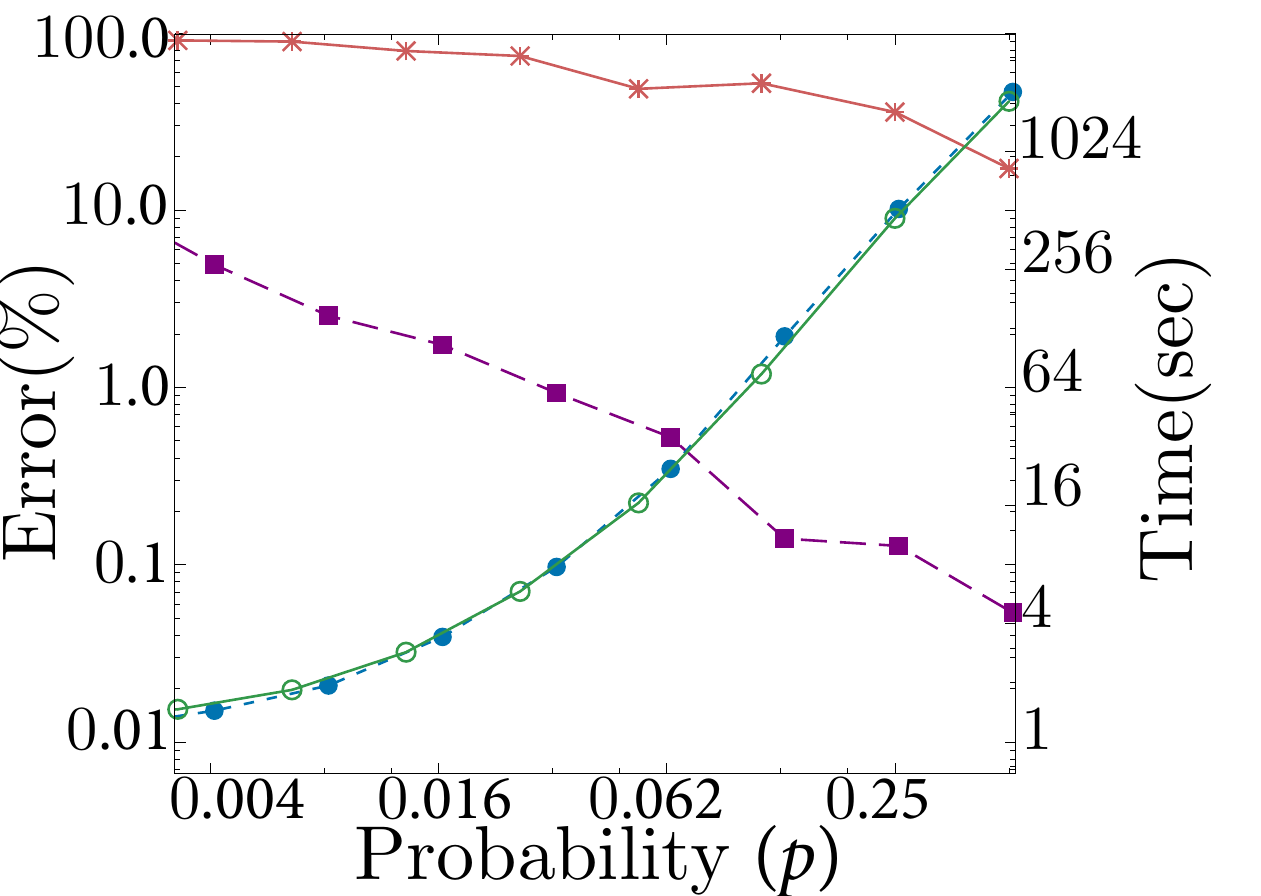}\label{figure:web-sprs-err-prob}}~{}~{}
	\hspace{1ex}
	\subfloat[][{\wiki}]{\includegraphics[width=0.33\textwidth]{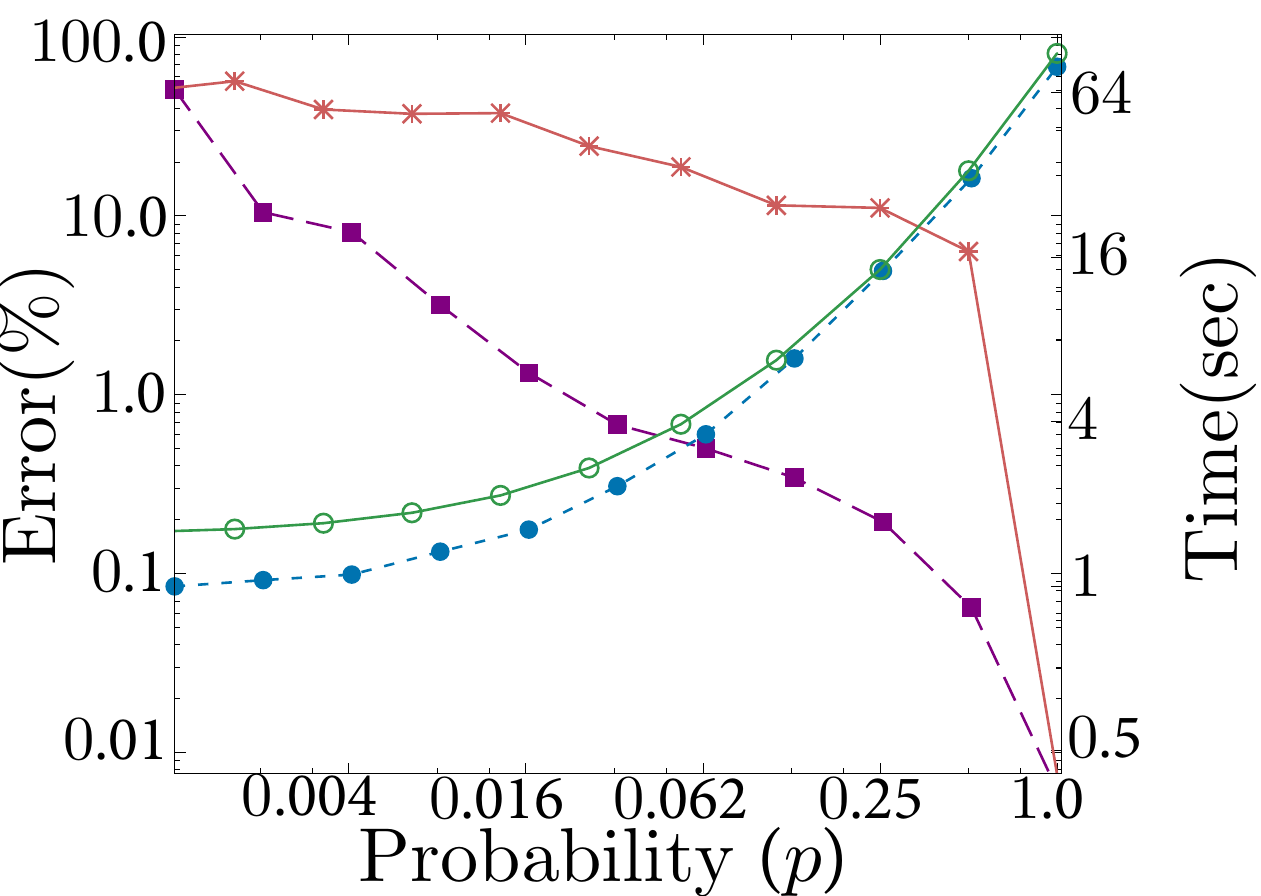}\label{figure:wiki-sprs-err-prob}}~{}~{}
	}
	\vspace{-1ex}
	\caption{\bf Accuracy (on left y-axis) and runtime performance (on right y-axis) of sparsification algorithms for different probabilities (on x-axis). \edgSpr~performs better than \clr, and yields $< 1\%$ relative error within $4$ seconds for all networks.}\label{fig:spr}
	%\erdem{replace those with the new 30 trial results and edit the text as needed. And please put pdfs of the figures. we can remove one of five for brevity. Also, put those in a proper subfigure format}}
	\vspace{-2ex}
\end{figure*}

%------------------------------------------------
\vspace{-1ex}
\subsection{Edge Sparsification (\edgSpr)}
\label{sec:edg-sparsify}
%------------------------------------------------
In \edgSpr, the input graph $G$ is sparsified by {\em independently} retaining each edge in $G$ in the sampled graph $G'$, with a probability $p$.  The number of butterflies in $G'$, obtained using \exact{}, is used to construct an estimate of $\bfly(G)$, after applying a scaling factor. This is much faster than working with the original graph $G$ since the number of edges in $G'$ can be much smaller. The \edgSpr~{}algorithm is described in \cref{algo:edge-spars}.

%%------------------------------------------------
%The analysis of \edgSpr has to deal with the complexity that even though the edges are sampled independently, the occurrence of different butterflies in the sample are not independent of each other. 

%------------------------------------------------
\begin{algorithm}[!t]
    \DontPrintSemicolon
    \SetKwInOut{Input}{Input}
    \Input{A bipartite graph $G = (V, E)$, parameter $p, 0 < p < 1$}
    Construct $E'$ by including each edge $e \in E$ independently with probability $p$\;
    
%    $E' \gets \emptyset$\;
%    \For {$e \in E$}
%    {
%            Toss a coin with probability $p$\;
%            \lIf{Coin is \textsc{Head}}
%            {
%                With probability $p$,$E' \gets E' \cup e$\;
%            }
%    }
   $\beta \gets \exact(V, E')$ using \cref{algo:exactBFC}\;
    \Return $\beta \times p^{-4}$\;
    \caption{\edgSpr: Edge Sparsification
    \label{algo:edge-spars}}
\end{algorithm}
%\vspace{-2ex}
%------------------------------------------------ 

%%------------------------------------------------
\begin{lem}\label{lem:edge-spars}
\vspace{-1ex}
Let $Y_{ES}$ denote the output of \textnormal{\edgSpr}~on input graph $G$. Then $\expec{Y_{ES}} = \bfly(G)$. If $p > \max\left(\sqrt[4]{\frac{24}{\bfly}}, \sqrt{\frac{24\Delta}{\bfly}}, \frac{24 \Delta^2}{\bfly} \right)$, then $\var{Y_{ES}} \le \bfly^2/8$.
%$\var{Y_{ES}}$ is also upper bounded by $\bfly{}p^{-4} + p_{1w}p^{-2} + p_{1e}p^{-1}$.
\vspace{-1ex}
\end{lem}

\begin{proof} For $i = 1, \cdots, \bfly$, let $X_i$ be a random variable indicating the $i^{th}$ butterfly,  s.t. $X_i$ is $1$ if the $i^{th}$ butterfly exists in the sampled graph $(V_G, E')$ and $0$ otherwise. Let $\beta$ be as defined in the algorithm. Clearly, $\beta = \sum_{i=1}^{\bfly} X_i$. We have $\expec{\beta} = \sum \expec{X_i} = \sum \prob{X_i=1}$. 
	%For a given butterfly $i$ to be retained in $G'$, we need all four of its edges to be sampled, which is true with probability $p^4$. 
	Since different edges are sampled independently, we have $\prob{X_i=1} = p^{4}$. 
	Hence, $\expec{Y_{ES}} = p^{-4} \expec{\beta} = p^{-4} \sum_{i=1}^{\bfly} p^4 = \bfly$.
	
	The random variables $X_i$s corresponding to different butterflies are not independent of each other, since butterflies may share edges. 
	Accounting for the covariances of $X_i, X_j$ pairs leads to:

	%\begin{footnotesize}
	
	\begin{equation*}
		\begin{aligned}
			\Var[Y_{ES}] &= \Var\left[p^{-4}\sum_{i=1}^{\bfly}{X_i}\right]  \\
			%          = p^{-8}\left[\sum_{i = 1}^{\bfly}{\Var[X_i]} + \sum_{i \neq j}{\Cov(X_i \wedge X_j)}\right] \\
			& = p^{-8}\left[\sum_{i = 1}^{\bfly}{(\E[X_i] - \E^2[X_i])} + \sum_{i \neq j}{\Cov(X_i \wedge X_j)}\right]\\
			&= p^{-8}\left[\bfly(p^4 - p^8) + \sum_{i \neq j}{(\E[X_iX_j] - \E[X_i]\E[X_j])}\right]
		\end{aligned}
	\end{equation*}
	
	%\end{footnotesize}
	Let $p_{0e}, p_{1e}$, and $p_{1w}$ respectively denote the number of pairs of butterflies that share zero edges, one edge, and one wedge (two edges). Note that these three cases account for all pairs of distinct butterflies, since it is not possible for two distinct butterflies to share three or more edges. Proceeding with a calculation similar to ones in \cref{sec:sampling}, we get:
	\begin{align*}
		\Var[Y_{ES}] &= p^{-8}\left[\bfly(p^4 - p^8) + p_{1e}(p^7 - p^8) + p_{1w}(p^6 - p^8)\right]\\
		&\leq ~{}\bfly{}p^{-4} + p_{1w}p^{-2} + p_{1e}p^{-1}
	\end{align*}
	
	%To compute $\sum_{i\neq{}j}{(\E[X_iX_j] - \E[X_i]\E[X_j])}$ we need to consider all possibilities for a pair of butterflies in the sampled graph.
	%If a pair of butterflies $i,j$ is of
	%\begin{itemize}
	%\item Type $0e$: $\E[X_iX_j] - \E[X_i]\E[X_j] = p^8 - p^8 = 0$
	%\item Type $1e$: $\E[X_iX_j] - \E[X_i]\E[X_j] = p^7 - p^8$
	%\item Type $1w$: $\E[X_iX_j] - \E[X_i]\E[X_j] = p^6 - p^8$
	%\end{itemize}
	%\begin{align*}
	%  \Var[Y_{ES}] &= p^{-8}\left[\bfly(p^4 - p^8) + p_{1e}(p^7 - p^8) + p_{1w}(p^6 - p^8)\right]\\
	%  &= ~{}\bfly{}p^{-4} + p_{1w}p^{-2} + p_{1e}p^{-1} - (\bfly + p_{1w} + p_{1e})\\
	%  &\leq ~{}\bfly{}p^{-4} + p_{1w}p^{-2} + p_{1e}p^{-1}
	%\end{align*}
	%------------------------------------------------

	%------------------------------------------------
	By Chebyshev's inequality we have $\Pr[|Y_{ES} - \bfly| \geq \epsilon\bfly] \leq \frac{\Var[Y_{ES}]}{\epsilon^2\E^2[Y_{ES}]}$. We need to find the probability $p$ such that $\Var[Y_{ES}] = o(\E^2[Y_{ES}])$ which results in $\E[Y_{ES}] = ~{}\bfly$ with probability $1 - o(1)$. Thus,
	%------------------------------------------------
	
	\begin{equation*}
		\begin{aligned}
			\E^2[Y_{ES}] &\gg \Var[Y_{ES}]  \\
			\bfly^2 &\gg ~{}\bfly{}p^{-4} + p_{1w}p^{-2} + p_{1e}p^{-1} - (\bfly + p_{1w} + p_{1e}) \\
			\bfly^2 &\gg ~{} \bfly{}p^{-4} + p_{1w}p^{-2} + p_{1e}p^{-1}
		\end{aligned}
	\end{equation*}
	
	%------------------------------------------------
	%Using \cref{lem:upperbounds}
	
	We know that
	$p > \max\left(\sqrt[4]{\frac{24}{\bfly}}, \sqrt{\frac{24\Delta}{\bfly}}, \frac{24 \Delta^2}{\bfly} \right)$.
	Using \cref{obs:pairs-ub}, this implies that:
	$p > \max\left(\sqrt[4]{\frac{24}{\bfly}}, \frac{\sqrt{24p_{1w}}}{\bfly}, \frac{24p_{1e}}{\bfly^2} \right)$, we get:
	\[
	\var{Y_{ES}} \le \bfly \cdot \frac{\bfly}{24} + p_{1w} \cdot \frac{\bfly^2}{24p_{1w}} + p_{1e} \cdot \frac{\bfly^2}{24p_{1e}} = \frac{\bfly^2}{8}
	\]
	%------------------------------------------------
	%\begin{align*}
	%\var{Y_{ES}} & \le \bfly p^{-4} + p_{1w}p^{-2} + p_{1e}p^{-1} \\
	%            & \le \bfly \cdot \frac{\bfly}{24} + p_{1w} \cdot \frac{\bfly^2}{24p_{1w}} + p_{1e} \cdot \frac{\bfly^2}{24p_{1e}}
	%             = \frac{\bfly^2}{8}
	%\end{align*}
	%------------------------------------------------
\end{proof}

%%------------------------------------------------
%%------------------------------------------------

\myremove{
The proof of the expectation is straightforward. For bounding the variance, the analysis has to deal with the fact that although different edges are sampled independently, the random variables corresponding to different butterflies being sampled are not independent of each other, since butterflies may share edges. By accounting for the covariances, we arrive that the following types of pairs of butterflies impact the variance: Type $0e$ (share zero edges), Type $1e$ (one edge), and Type $1w$ (one wedge). Let the numbers of such pairs be $p_{0e}, p_{1e}$, and $p_{1w}$ respectively. We arrive at a bound on the variance of $Y_{ES}$ using the following observation (whose proof is in~\cite{fullversion}), which allows a bound on the variance in terms of $\Delta$ and $\bfly$. 
}
% the bound on the variance is obtained by considering the
%co-variances of the random variables corresponding to different butterflies being sampled into the sparsified graph.

%
%

%We note the following bounds, proofs in the full version.
%------------------------------------------------
\begin{obs}
\vspace{-1ex}
%We further note the following upper bounds for the values of $p_{0e}$, $p_{1e}$, and $p_{1w}$ (proofs omitted due to space constraints):
$p_{2v} \leq \bfly{\triangle^2}$, $p_{1e} \leq \bfly\triangle^2$, and $p_{1w} \leq \bfly\triangle$.
\label{obs:pairs-ub}
\vspace{-1ex}
\end{obs}
Using Chebyshev's inequality and standard methods, the estimator $Y_{ES}$ can be repeated $\bigO(\log(1/\delta)/\epsilon^2)$ times to get an $(\epsilon,\delta)$ estimator of $\bfly(G)$.
%Using \cref{lem:edge-spars} along with Chebyshev's inequality, we get that $Y$ is a $(1/2,1/2)$-estimator of $\bfly(G)$.  We next use the standard method for improving the estimator: by taking the average of $\bigO(1/\epsilon^2)$ such estimates, we reduce the variance of get an $(\epsilon,1/2)$ estimator. By further taking the median of $\bigO(\log(1/\epsilon))$ estimators, we get an $(\epsilon,\delta)$ estimator of $\bfly(G)$.

%\remove{
%\red{The following terms upper bound the value of $p_{0e}$, $p_{1e}$, and $p_{1w}$:}
%    \begin{itemize}
%      %\item $p_{0e} \leq \frac{\bfly(e - 4)(e - 5)}{2}$
%      \item $p_{2v} \leq \bfly{\triangle^2}$
%      \item $p_{1e} \leq \bfly\triangle^2$
%      \item $p_{1w} \leq \bfly\triangle$
%    \end{itemize}
%}

%\begin{proof}
%The proof was omitted due to space constraint.
%\end{proof}

%------------------------------------------------
\subsection{Colorful Sparsification (\clr)}
\label{sec:col-sparsify}
%------------------------------------------------
The idea in \clr~is to sample edges at a rate of $p$, as in \edgSpr, but add dependencies between the sampling of different edges, such that there is a greater likelihood that a dense structure, such as the butterfly, is preserved in the sampled graph. The algorithm randomly assigns one of $\ncolors$ colors to vertices, and sample only those edges whose endpoints have the same color. The \clr~algorithm for approximate butterfly counting is presented in \cref{algo:clr-spars}.

We developed this method due to the following reason. Suppose $p = 1/N$. Though the expected number of edges in the sampled graph is $mp$, the same as in \edgSpr, it can be seen that the expected number of butterflies in the sampled graph is equal to $p^3\bfly$, which is higher than in the case of \edgSpr~{}($p^4 \bfly$). Thus, for a sampled graph of roughly the same size, we expect to find more butterflies in the sampled graph. Note however that this does not directly imply a lower variance of the estimator due to \clr.

%However, this does necessarily mean that an estimator based on \clr has a better variance and based on \edgSpr~ -- our calculations of variance show that in practice, \clr~ does not necessarily perform as well as \edgSpr.

%In the case of triangle counting, this yielded a good estimate for large graphs. %\erdem{we need to clarify our point here since exps show the reverse. I guess colorful sampling promises good estimates with less probability, but it suffers from the large variance(?)}
%----------------
\begin{lem}
\vspace{-1ex}
\label{lem:clr-spars}
Let $Y$ be the output of \textnormal{\clr} on input $G$, and $p = 1/\ncolors$.
$\expec{Y} = \bfly(G)$.
If ${}p > \max\bigl(\sqrt[3]{\frac{32}{\bfly}},\allowbreak \sqrt{\frac{32\triangle}{\bfly}}, \frac{32\triangle^2}{\bfly} \bigr)$, then $\var{Y} \le \bfly^2/8$.
%$\var{Y}$ is also upper bounded by $\bfly{}p^{-3} + p_{1w}p^{-2} + p_{1e}p^{-1} + p_{2v}p^{-1}$.
\vspace{-1ex}
\end{lem}

\begin{proof}
	For each butterfly $i =1,2,\ldots,\bfly(G)$ in $G$, let $X_i$ be a random variable such that $X_i = 1$ if the $i^\text{th}$ butterfly is monochromatic, i.e. all its vertices have been assigned the same color by the function $f$, and $X_i = 0$ otherwise. Clearly, the set of butterflies that are preserved in $E'$ are the monochromatic butterflies. Let $\beta$ be as defined in the algorithm. It follows that $\beta = \sum_{i=1}^{\bfly} X_i$, and $Y = p^{-3} \cdot \sum_{i=1}^{\bfly} X_i$.
4	
	Note that $\expec{X_i} = \prob{X_i=1} = p^3$. To see this, consider the vertices of butterfly $i$ in any order. It is only required that the colors of the second, third, and fourth vertices in the butterfly match that of the first vertex. Since vertices are assigned colors uniformly and independently, this event has probability $p^3$. Using linearity of expectation, $\expec{Y} = p^{3} \sum_{i=1}^{\bfly} \expec{X_i} = \bfly$.
	
	%The number of monochromatic butterflies $X$ is equal to the sum of these indicator variables, i.e., $X = \sum_{i = 1}^{\bfly}{X_i}$.\\
	%By linearity of expectation and the fact that $\Pr[X_i = 1] = p^3$, we get
	%By linearity of expectation, we have
	%$\E[Y] = \sum_{i = 1}^{\bfly}{\E[X_i]} = p^3\cdot\bfly $. Hence,
	%$\bfly~{}= p^{-3}\cdot \E[X]$
	
	%We know $\E[Y] = p^{-3}\sum_{i = 1}^{\bfly}{\E[X_i]} =~{}\bfly$.
	
	For its variance, using the same argument in \cref{algo:edge-spars}, we have:
	\begin{small}
		\begin{equation*}
			\begin{aligned}
				\Var[Y] %&= \Var[p^{-3}\sum_{i=1}^{\bfly}{X_i}] \\
				%&= p^{-6}\left[\sum_{i = 1}^{\bfly}{\Var[X_i]} + \sum_{i \neq j}{\Cov(X_i \wedge X_j)}\right] \\
				%&= p^{-6}\left[\sum_{i = 1}^{\bfly}{(\E[X_i] - \E^2[X_i])} + \sum_{i \neq j}{\Cov(X_i \wedge X_j)}\right]\\
				&= p^{-6}\left[\bfly(p^3 - p^6) + \sum_{i \neq j}{(\E[X_iX_j] - \E[X_i]\E[X_j])}\right]
			\end{aligned}
		\end{equation*}
	\end{small}

	For a pair of butterflies $(i,j)$ of.
	\begin{itemize}
		\itemsep0em 
		\item
		Type $0v$ or $1v$: $\expec{X_iX_j} - \expec{X_i}\expec{X_j} = p^6 - p^6 = 0$
		
		\item
		Type $2v$: $\expec{X_iX_j} - \expec{X_i}\expec{X_j} = p^5 - p^6$
		
		\item
		Type $1e$: $\expec{X_iX_j} - \expec{X_i}\expec{X_j} = p^5 - p^6$
		
		\item
		Type $1w$: $\expec{X_iX_j} - \expec{X_i}\expec{X_j} = p^4 - p^6$
	\end{itemize}
	
	Therefore, we have
	\begin{small}
		\begin{equation*}
			\begin{aligned}
				\Var[Y] &= p^{-6}\left[\bfly(p^3 - p^6) + p_{2v}(p^5 - p^6) \right.\\
				& \left. + p_{1e}(p^5 - p^6) + p_{1w}(p^4 - p^6)\right]\\
				&= ~{}\bfly{}p^{-3} + p_{1w}p^{-2} + p_{1e}p^{-1} + p_{2v}p^{-1} \\
				&- (\bfly + p_{1w} + p_{1e} + p_{2v}) \\
				&\leq ~{}\bfly{}p^{-3} + p_{1w}p^{-2} + p_{1e}p^{-1} + p_{2v}p^{-1}
			\end{aligned}
		\end{equation*}
	\end{small}
	Using Observation \ref{obs:pairs-ub}, $p > \max\left( \sqrt[3]{\frac{32}{\bfly}}, \sqrt{\frac{32\triangle}{\bfly}}, \frac{32\triangle^2}{\bfly} \right)$ implies that:\\  $p > \max\left( \sqrt[3]{\frac{32}{\bfly}}, \frac{\sqrt{32p_{1w}}}{\bfly}, \frac{32p_{1e}}{\bfly^2}, \frac{32p_{2v}}{\bfly^2} \right)$. Hence,
	\begin{small}
		\begin{equation*}
			\begin{aligned}
				\var{Y} & \le \bfly p^{-3} + p_{1w}p^{-2} + p_{1e}p^{-1} + p_{2v}p^{-1} \\
				& \le \bfly \cdot \frac{\bfly}{32} + p_{1w} \cdot \frac{\bfly^2}{32p_{1w}} + p_{1e} \cdot \frac{\bfly^2}{32p_{1e}} + p_{2v} \cdot \frac{\bfly^2}{32p_{2v}}
				= \frac{\bfly^2}{8}
			\end{aligned}
		\end{equation*}
	\end{small}
\end{proof}
\begin{algorithm}[!t]
    \DontPrintSemicolon
    \SetKwInOut{Input}{Input}
    \Input{Bipartite graph $G = (V, E)$, number of colors $\ncolors$}
%    Number of colors $\ncolors = \frac{1}{p}$\;
%	$p \gets 1/\ncolors$\;
    Let $f : V \to \{1,\ldots,\ncolors\}$\tcp*{map to random colors}%\;
    $E' \gets \{(u, v) \in E_G | f(u) = f(v)\}$\;
    $\beta \gets$ \exact~$(V, E')$~\tcp*{\cref{algo:exactBFC}}
    \Return $\beta \cdot p^{-3}$ where $p=1/\ncolors$\;
    \caption{\clr
    \label{algo:clr-spars}}
\end{algorithm}
%------------------------------------------------

\begin{table}[t] 
	\vspace{1ex}
	\centering
	\resizebox{\textwidth}{!}{
		\begin{tabular}{|c| @{\extracolsep{\fill}} |c|c|}
			\hline
			\multirowcell{2.3}[1.9ex]{\diagbox[height=1.95\line]{\rlap{\enspace\raisebox{0.6ex}{Graph}}}{\raisebox{-0.5ex}{Algorithm}}} & \makecell{\edgSpr{}\\$\bfly{}p^{-4} + p_{1w}p^{-2} + p_{1e}p^{-1}$} & \makecell{\clr{}\\$\bfly{}p^{-3} + p_{1w}p^{-2} + p_{1e}p^{-1} + p_{2v}p^{-1}$} \\\hline
			\deli & $4.047 \times 10^{15}$ & $9.163 \times 10^{20}$ \\ \hline
			\jrn & $1.413 \times 10^{19}$ & $2.042 \times 10^{25}$ \\ \hline
			\ork & $5.576 \times 10^{19}$ & $8.514 \times 10^{25}$ \\ \hline
			\web & $7.760 \times 10^{20}$ & $4.692 \times 10^{27}$  \\ \hline
			\wiki & $1.303 \times 10^{20}$ & $3.062 \times 10^{26}$\\ \hline
	\end{tabular}}
	\caption{Upper bounds on \edgSpr~\& \clr~variances, $p$=$0.1$.}\label{table:variance-sprs}
	\vspace{-1ex}
\end{table}

\vspace{-2ex}
%----------------------------------------------------------------------------------------------------
\subsection{Comparison of Sparsification Algorithms}
%----------------------------------------------------------------------------------------------------

The parameter $p$ controls the probability of an edge being included in the sparsified graph. As $p$ increases, we expect the accuracy as well as the runtime to increase. 
%We note that both sparsification methods, \edgSpr~{}and \clr~{}produce unbiased estimates for $\bfly(G)$. 
The relative accuracies of the two methods depends on the variances of the estimators. 
We estimated the variances using our analysis, and \cref{table:variance-sprs} presents the results.
We observe that the variance of \edgSpr~{}is predicted to be much lower than that of \clr. 
To understand this, we begin with \cref{lem:edge-spars,lem:clr-spars} which show expressions bounding the variances in terms of $p_{1w}, p_{1e},$ and $p_{2v}$, also summarized in \cref{table:variance-sprs}. The difference in the variance between \clr~ and \edgSpr~ boils down to $(\bfly{}p^{-3}+p_{2v}p^{-1}-\bfly{}p^{-4})$. In our experiments, we found that $p_{2v}$, the number of pairs of butterflies that share two vertices without sharing an edge, was very high, much larger than $p_{1e}$, $p_{1w}$, and $\bfly$. This explains why the variance of \edgSpr~was higher. 

This observation about variance is consistent with our experimental results.
In \cref{fig:spr}, we report the relative percent error as well as the runtime as the sampling probability $p$ increases. Results show that \edgSpr~{}obtains less than one percent error when $5$ percent of edges are sampled. However, as shown in \cref{figure:jrn-sprs-err-prob,figure:web-sprs-err-prob,figure:wiki-sprs-err-prob}, \clr~{}requires a larger sampling probability to achieve a reasonable accuracy. 
%This is in line with the theoretical accuracy predictions based on the variance.

%---------------------------------------------------------------------
%\subsection{Comparison of Sparsification Methods}
%Both edge and colorful sparsification methods produce unbiased estimates of $\bfly(G)$. The density of butterflies in the sampled graph from \clr~{}is higher than in the case of \edgSpr. But which method yields a better accuracy?

%To answer this, we first compare the theoretical results for the two methods. We note that the derived lower bounds sampling probability required for a good variance for the two methods (\cref{lem:edge-spars,lem:clr-spars}) are quite similar and the one required for colorful is slightly better than the one required for \edgSpr~{}($\bfly^{-1/3}$ as opposed to $\bfly^{-1/4}$). Note that a smaller lower bound on the probability is better since the size of the sampled graph is smaller.

%However, in the experiments, we found that \edgSpr~{}performs considerably better than \clr. {\color{red} Explain this discrepancy.}

%-------------------
\begin{table}[!b] 
%	\vspace{-1ex}
	\centering
	\resizebox{0.6\textwidth}{!}{
		\begin{tabular}{|c| @{\extracolsep{\fill}} |c|c|}
			\hline
			 & \edgSam~(with \fstEdg) & \edgSpr \\\hline
			\deli & $3.4$ & $2.1$ \\ \hline
			\jrn  & $5.0$ & $1.7$ \\ \hline 
			\ork & $3.4$  & $3.4$ \\ \hline
			\web  & $4.1$ & $3.9$  \\ \hline
			\wiki & $4.8$ &$2.3$ \\ \hline
	\end{tabular}}
%	\vspace{1ex}
	\caption{Time (in seconds) to obtain 1\% relative percent error for the best sampling and sparsification algorithms.%\erdem{\edgSpr~for web is not fixed, Vahid said the correct run gives 3.977 secs for the first time it goes below 1\% error -- all numbers are being updated now.}
	}\label{table:bestofboth}
	%\edgSpr~is faster on all graphs, except \web.\erdem{we'll fix this hopefully
	\vspace{-5ex}
\end{table}
%-------------------

\subsection{Sampling or Sparsification?}

\noindent %We now compare the best of each set of algorithms. 
The accuracy of the best sampling algorithm, \edgSam~with \fstEdg, is compared with the best sparsification algorithm, \edgSpr, in~\cref{table:bestofboth}.
Overall, two algorithms take similar times to reach a 1\% error on all graphs we considered, in the range of 1.7-5 sec, with \edgSpr~achieving this accuracy faster than \edgSam~with \fstEdg.

However, \edgSpr~has other downsides when compared with \edgSam.
First, the memory consumption of \edgSpr~is $O(mp)$ where $p$ is a parameter, and is larger than \edgSam~with \fstEdg , whose memory consumption is $O(\Delta)$. As a result, we expect the memory of \edgSpr~to be linearly in the size of the graph, showing that it may be easier for \edgSam~to scale to graphs of even larger sizes. Next, \edgSpr~needs to decide on a sampling parameter $p$ to balance between accuracy and runtime. If $p$ is too large, then the runtime is high, and if $p$ is too small, then the accuracy is low. Finally, one-shot sparsification algorithms assume that the entire graph is available whereas the \edgSam~needs access to only a subgraph of the graph. Thus if one has to pay for data about the graph, or if the data about edges is hard to obtain, then sampling may be the better alternative.

Overall, local sampling algorithms can find a larger application space in real-world.
{\em If the entire graph is available, and memory is not a bottleneck, it may be better to go with \edgSpr.
For more restrictive scenarios, \edgSam~combined with \fstEdg~is the better option.}

%Here we investigate whether continuous sampling is even better than the one-shot sparsification.
%We compare the \edgSam~algorithm that uses \fstEdg~vs. the \edgSpr~algorithm with respect to the runtimes needed to go below $1\%$ error.
%\cref{table:bestofboth} presents the results.

\section{Conclusion}
We introduced a suite of algorithms for butterfly counting in bipartite networks. We first showed that a simple statistic about vertex sets, which is cheap to obtain, helps drastically to reduce the runtime of exact algorithms. We then presented scalable randomized algorithms that approximate the number of butterflies in a graph, with provable accuracy guarantees. 

Randomized algorithms using one-shot sparsification are applicable when the entire graph is available locally, and rely on a global sampling step to compute a smaller ``sparsified'' subgraph, which is used to compute an accurate estimate. On the other hand, if the access to the graph data is limited, local sampling algorithms can be used to randomly sample small subgraphs of the entire graph, and analyze them to compute an estimate. Sampling algorithms are especially beneficial when there is a rate-limited API that provides random samples of the network data, such as the GNIP~\cite{gnip} and Facebook Graph API~\cite{fbgraph}. Our best sampling and sparsification algorithms yield less than $1\%$ relative error within $\approx 5$ and $\approx 4$ seconds for all the networks we considered, whereas the state-of-the-art exact algorithm does not complete even in 40,000 secs on the \web~graph.

There are many promising directions to follow. Here we only focused on the butterfly motif, but one can consider other motifs suitable to bipartite graphs.
Note that the combinatorial explosion for the number of butterflies is more serious than the triangles in unipartite graphs.
Thus it is challenging to consider even larger motifs, but the ideas in this work can serve as building blocks for considering different motifs. 
Another direction is adapting our algorithms for the streaming scenario. Given a single pass over the graph, the question is how to sample/sparsify the graph stream to accurately estimate the number of butterflies. Sparsification-based algorithms may be adapted to the streaming scenario quite easily, but same cannot be said for the local sampling algorithms.
%It is also interesting to pursue space-bounded one-pass bipartite graph motif counting in greater detail.
%The promising performance of the edge sampling over the other alternatives suggests that randomized streaming algorithms would be in a similar nature.

\bibliographystyle{plain}
\bibliography{paper}

\end{document}